\documentclass[sigconf]{acmart}
\usepackage{color,colortbl}
\usepackage{cleveref}

\crefformat{section}{\S#2#1#3} 
\crefformat{subsection}{\S#2#1#3}
\crefformat{subsubsection}{\S#2#1#3}

\usepackage{multirow}
\usepackage{hhline}
\usepackage{amssymb}
\usepackage{array}
\usepackage{pifont}
\usepackage{enumerate}

\usepackage{graphicx,epsfig,amsfonts,bbm,epstopdf,tabularx}
\usepackage{subfigure}

\usepackage{latexsym}
\usepackage{url}

\usepackage{xcolor}
\usepackage{algorithm,algpseudocode}
\usepackage{algorithmicx}
\usepackage{stmaryrd}
\usepackage{bm}
\usepackage{multirow}
\usepackage{graphicx}
\usepackage{arydshln}
\usepackage{mathtools}
\usepackage{tabulary}
\usepackage{booktabs}
\usepackage{subfigure}
\usepackage{xspace}

\newtheorem{myDef}{Definition}
\newtheorem{myObs}{Observation}
\algnewcommand{\LeftComment}[1]{\Statex \(\triangleright\) #1}

\newcolumntype{L}[1]{>{\raggedright\let\newline\\\arraybackslash\hspace{0pt}}m{#1}}
\newcolumntype{C}[1]{>{\centering\let\newline\\\arraybackslash\hspace{0pt}}m{#1}}
\newcolumntype{R}[1]{>{\raggedleft\let\newline\\\arraybackslash\hspace{0pt}}m{#1}}

{\hspace*{\fill}$\Box$\par\vspace{4mm}}

\newboolean{showcomments}
\setboolean{showcomments}{true}

\newcommand{\mohammad}[1]{  \ifthenelse{\boolean{showcomments}}
	{\textcolor{red}{(Mohammad says:  #1)}}{}}

\ifodd 1
\newcommand{\com}[1]{\textbf{\color{blue} (COMMENT: #1)}}
\else

\newcommand{\com}[1]{}
\fi

\newcommand{\be}{\begin{equation}}
\newcommand{\ee}{\end{equation}}
\def\bee#1\eee{\begin{align}#1\end{align}}
\newcommand{\bse}{\begin{subequations}}
	\newcommand{\ese}{\end{subequations}}

\newcommand{\ocm}{\texttt{BatMan}\xspace}
\newcommand{\ocmrate}{\texttt{BatManRate}\xspace}

\newcommand{\prob}{\textsc{OLIM}\xspace}

\newcommand{\probeqs}{\textsc{OLIM}\xspace}
\newcommand{\ofa}{\texttt{OPT}\xspace}

\newcommand{\fon}{\texttt{OnFix}\xspace}
\newcommand{\lyp}{\texttt{LypOpt}\xspace}
\newcommand{\bon}{\texttt{NoSTR}\xspace}
\newcommand{\preday}{\texttt{PreDay}\xspace}

\newcommand{\costopt}{\textsf{cost}(\ofa)\xspace}
\newcommand{\costocm}{\textsf{cost}(\ocm)\xspace}
\newcommand{\costocmrate}{\textsf{cost}(\ocmrate)\xspace}

\newcommand{\dt}{d(t)\xspace}

\newcommand{\xt}{x(t)\xspace}

\newcommand{\xit}{x_{i}(t)\xspace}

\newcommand{\bt}{b(t)\xspace}
\newcommand{\bd}{b_{\texttt{d}}(t)\xspace}

\newcommand{\pt}{p(t)\xspace}

\makeatletter
\def\@copyrightspace{\relax}
\makeatother

\copyrightyear{2019}
\acmYear{2019}
\setcopyright{rightsretained}
\acmConference{Submitted}{2019}

\begin{document}
	%
	\title{Online Inventory Management \\
	with Application to Energy Procurement in Data Centers}


	\author{Lin Yang}
	\authornote{Both authors contributed equally to this research.}
	\affiliation{%
	\institution{Chinese University of Hong Kong}}
	\email{yl015@ie.cuhk.edu.hk}

	\author{Mohammad H. Hajiesmaili}
	\authornotemark[1]
	\affiliation{%
		\institution{UMass Amherst}}
	\email{hajiesmaili@cs.umass.edu}
	
	\author{Ramesh Sitaraman}
	\affiliation{%
		\institution{UMass Amherst and Akamai}}
	\email{ramesh@cs.umass.edu}
	
	\author{Enrique Mallada}
	\affiliation{%
		\institution{Johns Hopkins University}}
	\email{mallada@jhu.edu}
	
	\author{Wing S. Wong}
	\affiliation{
		\institution{Chinese University of Hong Kong}}
	\email{wswong@ie.cuhk.edu.hk}
	
	\author{Adam Wierman}
	\affiliation{%
		\institution{California Institute of Technology}}
	\email{adamw@caltech.edu}
	
	
	%


\renewcommand{\shortauthors}{Yang and Hajiesmaili, et al.}	
\renewcommand{\shorttitle}{Online Inventory Management ...}		
	
\begin{abstract}
	Motivated by the application of energy storage management in electricity markets, this paper considers the problem of online linear programming with inventory management constraints.  Specifically, a decision maker should satisfy some units of an asset as her demand, either form a market with time-varying price or from her own inventory. The decision maker is presented a price in slot-by-slot manner, and must immediately decide the purchased amount with the current price to cover the demand or to store in inventory for covering the future demand. The inventory has a limited capacity and its critical role is to buy and store assets at low price and use the stored assets to cover the demand at high price. The ultimate goal of the decision maker is to cover the demands while minimizing the cost of buying assets from the market. 
	
	We propose \ocm, an online algorithm for simple inventory models, and \ocmrate, an extended version for the case with rate constraints. Both \ocm and \ocmrate achieve optimal competitive ratios, meaning that no other online algorithm can achieve a better theoretical guarantee.
	To illustrate the results, we use the proposed algorithms to design and evaluate energy procurement and storage management strategies for data centers with a portfolio of energy sources including the electric grid, local renewable generation, and energy storage systems.

\end{abstract}

	
	
	%
	\maketitle

\section{Introduction}
Online optimization and decision making under uncertainty is a fundamental topic that has been studied using a wide range of theoretical tools and in a broad set of applications. On the theoretical side, it has been approached from the perspective of competitive algorithms with the competitive ratio as the performance metric~\cite{Borodin98}, online learning with the regret as the performance metrics~\cite{hazan2016introduction,yang2018optimal}, and reinforcement learning with different modeling techniques such as Markov decision process~\cite{sutton2018reinforcement}. On the application side, recent scenarios where theoretical results have had an impact for real-world design include data center optimization~\cite{tu2013dynamic,lin2013dynamic,albers2017energy,ren2018datum,albers2018optimal}, energy systems~\cite{zhang2018peak,zhou2015online,khezeli2018risk,hajiesmaili2017crowd}, cloud management~\cite{mao2016optimal,zhang2017optimal,kobayashi2007tight},  computer and communication networks~\cite{feldkord2018online,chan2007online,hajiesmaili2017incentivizing,even2018online}, and beyond.

Motivated by storage management problem for data centers procuring energy from the electricity market, this paper studies a generalization of the classical online optimization formulation: \textit{online linear programming with inventory management} (\prob). In this problem, in each slot, a decision maker should satisfy $\dt$ units of an asset, e.g., energy demand, as her demand, either from a market with time-varying price or from her own inventory, e.g., energy storage system. In slot $t$, the decision maker is presented a price $\pt$, and must decide $\xt$ as the procurement amount with the current price to cover the demand $\dt$ or additionally to store in inventory for covering the future demand. The capacity of inventory is $B$ units, and its critical role is to buy and store assets at low price and use those for covering the demand at high price.
The goal is to cover the demands while minimizing the cost of buying from the market. A formal statement of the problem is presented in \S~\ref{sec:formulation}.


\probeqs captures a variety of timely applications in different domains, e.g., booking hotels or flight tickets in advance to for high season by travel agencies or charging energy storage in data centers to use in high price periods.
The application that motivates our interest in studying \probeqs is designing energy procurement and storage management strategies for large-scale electricity customers such as data centers, university campuses, or enterprise headquarters. Usually, these customers can satisfy their energy demand from a portfolio of sources, including the grid, local renewable sources, and on-site energy storage systems. Notable examples are thermal energy storage in Google data center in Taiwan~\cite{taiwan}, Tesla batteries to power Amazon data center in California~\cite{amazontesla}, Google data center with on-site renewable sources in Belgium~\cite{Belgium}, and large-scale batteries in Apple Park's microgrid~\cite{applebattery}. The electricity pricing for large customers is moving toward real-time pricing and the price changes dynamically over time~\cite{zhang2018peak,ghamkharienergy,xu2014reducing}.
The addition of on-site storage systems presents a great opportunity for shifting energy usage over time to reduce the energy cost by purchasing the energy at low price periods.
The severe uncertainty in energy demand and price, however, make designing optimal energy procurement strategies a challenging task and emphasizes the need for \textit{online solution design.} In \S\ref{sec:uncertainty}, we present the detailed energy procurement scenario and the challenges due to uncertainty in the problem.

Note that \prob is a generalization of several classic online algorithmic problems. The first category is online search for the optimum with sequential arrival of online price, known as online search or conversion problems~\cite{mohr2014online}. Notable examples are the time series search and one-way trading problems~\cite{Yaniv01}, the multiple-choice secretary problem \cite{babaioff2008online}, the $k$-search problem~\cite{lorenz2009optimal}, and online linear programs with covering constraints~\cite{buchbinder2009online}. Different from these problems, in \prob, in addition to the uncertainty in the market pricing, we have another uncertainty due to online arrival of the demand.
In other words, \prob comes with two sets of uncertain input parameters, which allow the adversary to have more options in constructing the worst-case input. In terms of constraints, \prob includes inventory management constraints that couples the covering constraints over time. More details are given in~\S\ref{sec:formulation}.



\textbf{Summary of Contributions.}
In this paper, we develop a online algorithms for \prob and show that the algorithm achieves the minimal competitive ratio achievable by an online algorithm.
More specifically,  we propose \ocm\footnote{\ocm is short for Battery Management, inspired from our application of interest in optimizing energy procurement by \textit{battery (energy storage) management}.} (\S\ref{sec:batman}) and \ocmrate (\S\ref{sec:batmanrate}). \ocm is a simpler algorithm that works for the inventory that have no rate constraints, i.e., no limit on input and output rate to/from the inventory at any slot.
\ocmrate works in a more general context where the input and output (a.k.a., charge and discharge, in the application context) rates are bounded.

The high-level intuition behind the design of \ocm is to store assets at cheap and use the stored asset once the price is expensive. However, the dynamic pricing and dynamic demand make this decision making challenging. The main ideas of \ocm are: (i) \textit{adaptive reservation based on storage\footnote{Throughout the paper, \textit{inventory} and \textit{storage} are used interchangeably.} utilization} that tackles the challenges due to price dynamics; and (ii) \textit{construction of virtual storages} that tackle the challenges due to demand dynamics.
The idea of using adaptive pricing function is adapted from the online algorithms for $k$-search problem in~\cite{lorenz2009optimal}.

\textit{The main novelty in the algorithm design is introducing the novel notion of virtual storages to tackle the additional demand uncertainty.} In particular, given some back-up assets in inventory, one can see satisfying the demand in each slot as as the buying (minimization) version of an optimal search problem with the current demand as the target amount. However, dynamic arrival of demands makes these online search problems coupled over time, and exacerbates the competitive analysis of the algorithms.
%

Our main technical results provide an analysis of the competitive ratios of \ocm and \ocmrate. These results are summarized in Theorems~\ref{thm:cr} and~\ref{thm:batmanrate}.

\begin{theorem}
	\label{thm:cr}
	With the following reservation function
	\begin{equation}
	\label{eq:g_fun}
	G_{B_i}(p)=\alpha B_i\ln \left[\left(1-\frac{p}{p_{\max}}\right)\frac{\alpha}{\alpha-1}\right],~p\in\left[p_{\min},\frac{p_{\max}}{\alpha}\right],
	\end{equation}
	\emph{\ocm} achieves the optimal competitive ratio of $\alpha$ defined as
	\begin{equation}
	\label{eq:alpha}
	\alpha = \left(W\left(-\frac{\theta -1}{\theta\exp(1)}\right)+1\right)^{-1}.
	\end{equation}
	
\end{theorem}



In Equation~\eqref{eq:alpha}, $\theta$ is the price fluctuation ratio, and $W(.)$ is Lambert-W function, defined as the inverse of $f(z) = z\exp(z)$. 

\textit{Interestingly, the above competitive ratio for \ocm is exactly the same as the optimal competitive ratio for $k$-min search problem (see~\cite[Theorem~2]{lorenz2009optimal}), when $k \rightarrow \infty$}. 
However, \prob involves additional uncertainty on demand as compared to~\cite{lorenz2009optimal} and additional inventory management constraint. The additional demand uncertainty enlarges the design space of the adversary and complicates the competitive analysis. To obtain the performance bounds of online and offline algorithms, we introduce several novel techniques and notions, such as definition of reservation and idle periods. 


\begin{theorem}
	\label{thm:batmanrate}
	\ocmrate achieves the optimal competitive ratio of $\alpha$ as in Equation~\eqref{eq:alpha}.
\end{theorem}

\ocmrate extends \ocm to the case with rate constraints. Two significant changes in algorithm design are adaptive determination of the capacity of virtual storages to reflect the output rate, and adaptive setting of reservation price to reflect the input constraints. While the general logic for the analysis of \ocmrate is similar to that of \ocm, it comes with a significant result on showing that in worst-case, the output rate constraint is not active.

In addition to providing theoretical analysis of \ocm, in \cref{sec:exp}, we also empirically evaluate the performance of the proposed algorithms using real-world data traces in the data center energy scenario. More specifically, evaluate our algorithms using extensive data traces of electricity prices from several electricity markets (CAISO~\cite{CAISO}, NYISO~\cite{NYISO}, ERCOT~\cite{ERCOT}, DE Market~\cite{GERMAN}), energy demands from multiple data centers of Akamai's CDN~\cite{nygren2010akamai}, and renewable production values from solar~\cite{dobos2014pvwatts} and wind installations~\cite{wind,wind2}. In a broad set of representative scenarios that include different seasons and locations, \ocm achieves a cost reduction of 15\% in comparison with using no energy storage at all, establishing the value of batteries for energy procurement. Further, \ocm achieves an energy procurement cost that is within 21--23\%  of the theoretically smallest achievable cost in an offline setting. In addition,  \ocm outperforms state-of-the-art algorithms for energy procurement by 7.5--10\%. Finally, \ocmrate outperforms all the alternatives in our experiments.



\section{Problem Formulation}
\label{sec:formulation}
We present the system model and formulate the problem. The interpretation of the  model parameters is illustrated by the application of storage-assisted energy procurement in electricity market. 

We assume that the time-slotted model in which the time horizon is divided into $T$ slots, indexed by $t$, each with fixed length, e.g., $5$ minutes in California ISO (CAISO) and New York ISO (NYISO)~\cite{epri2016}. 
We consider the following scenario. At each slot, a demand $\dt$ arrive online that must be satisfied from either the market with the real-time price $\pt$ or from the local inventory, i.e., energy storage system. The decision maker can purchase more from the market and store in the inventory to satisfy the future demand. The ultimate goal is to design an algorithm to determine the value of $\xt$, as the procurement amount in each slot, such that the procurement cost is minimized over a time horizon, and the demand is satisfied. 
In the following, we introduce the inventory management constraints.

\paragraph{Inventory Management Constraints}
Let $B$, $\rho_c$, and $\rho_d$ be the capacity, the maximum input rate, and the maximum output rate of the inventory. Let $\bt \in [0, B]$ be the inventory level at the end of slot $t$ that represents the amount of assets that are already in the inventory. The evolution of inventory level is given by
\begin{equation}
\label{eq:et}
\bt= b(t-1) + \xt - \dt,
\end{equation}
which states that the amount of assets in the inventory at the end of each round is equal to the previous existing amount $\b(t-1)$ and the current procurement amount $\xt$ subtracted by the demand $\dt$. 
Moreover, we have two constraints on the value of $\xt$:
\begin{equation}
\xt \geq \dt - \min \big\{\rho_d, b(t-1)\big\}, \label{eq:discharge}
\end{equation}
which captures the maximum output rate from the inventory and ensures covering the demand, and
\begin{equation}
\xt \leq \dt + \min \big\{\rho_c, B - b(t-1)\big\}, \label{eq:charge}
\end{equation}
which captures the input rate constraint to the inventory. 
Finally, we have the following inventory capacity constraint
$$0 \leq \bd \leq \min \{\rho_d, b(t-1)\},$$

In our example of energy storage, the rate constraints $\rho_c$ and $\rho_d$ are the charge and discharge rate of the energy storage systems. Several examples of actual values of $\rho_c$ and $\rho_d$ for different energy storage technologies, e.g., lead-acid and lithium-ion batteries, and compressed air energy storage are provided in \S\ref{sec:exp_batmanerate}.

\paragraph{Problem Formulation}
We can now summarize the full formulation of online optimization with inventory management (\prob).
If the demands and market prices are known for the entire time horizon in advance, the {\em offline} version of \prob  can be formulated a linear program as follows.
\begin{eqnarray}
\label{eq:ep2eq}
\prob: & \min &  \sum_{t\in\mathcal{T}} \pt \xt \nonumber\\
& \mathrm{s.t.}: & \forall t\in \mathcal{T}: \nonumber\\
&& \xt \geq \dt - \min \big\{\rho_d, b(t-1)\big\}, \label{eq:dt_cover}\\
&& \xt \leq \dt + \min \big\{\rho_c, B - b(t-1)\big\}, \label{eq:dt_pack}\\
&&\bt=b(t-1)+ \xt - \dt, \label{eq:sl}\\
&& 0 \leq \bt \leq B, \label{eq:bat_cap}\\
&\mathrm{vars.}:& \xt \in \mathbb{R}^+. \nonumber
\end{eqnarray}

The objective is to minimize the procurement cost from the market. Constraints~\eqref{eq:dt_cover}-\eqref{eq:dt_pack} ensure covering the demand and rate limits.
Constraint~\eqref{eq:sl} dictates the evolution of the inventory, and~\eqref{eq:bat_cap} enforce the capacity of the inventory.
Since \prob is a linear program it can be solved efficiently in an offline manner. 

In this work, we are interested in developing {\em online algorithms} for \prob that make decisions at each time $t$, knowing the past and current prices and demands, but not knowing those same inputs for the future. The algorithmic challenge is to procure assets from the market and store in the storage in the current time, without knowing if such decisions will work out favorably in the future. The classical approach for evaluating online algorithms is competitive analysis, where the goal is to design algorithms with the smallest {\em competitive ratio}, that is, the cost ratio between the online algorithm and an {\em offline} optimal algorithm that has access to {\em complete} input sequence. In our work, we devise an online algorithm that has provably the best competitive ratio for \prob.
In the design of algorithms, we assume that the values of $p_{\max}$ and $p_{\min}$, as the maximum and minimum prices, are known a priori. This assumption is reasonable since by the historical data, these values could be predicted. Further, related problems makes the similar assumptions~\cite{lorenz2009optimal,mohr2014online,Yaniv01}. Let $\theta = p_{\max}/p_{\min}$ as the price fluctuation ratio. Our analysis characterizes the performance as a function of $\theta$.

\paragraph{The Case Study}
\label{sec:uncertainty}
Our motivation for studying \prob comes from energy storage management in electricity markets.  In such scenarios the demand and price values are highly uncertain and unpredictable. 
To further motivate the online algorithm design, in the following, we demonstrate the uncertainty of these values using real data-traces from electricity markets and Akamai data centers. 

1) The electricity pricing for large customers like data centers is moving toward real-time pricing and the price changes dynamically over time~\cite{urgaonkar2011optimal,huang2012optimal,guo2012TPDS,van2013optimal,chau2016cost}. 
Two examples of real-time energy prices in NYISO and DE Market are demonstrated in Figure~\ref{fig:market_prices}. By comparing the price dynamics in two different electricity market, we can see totally different patterns. While the prices in NYISO highly fluctuate  without any regular pattern, in German Electricity Market, we observe regular daily patterns with low price fluctuations.

\begin{figure}[!h]
	\begin{center}
		\subfigure[NYISO]{\label{fig:ny_prices}
			\includegraphics[angle=0,scale=0.18]{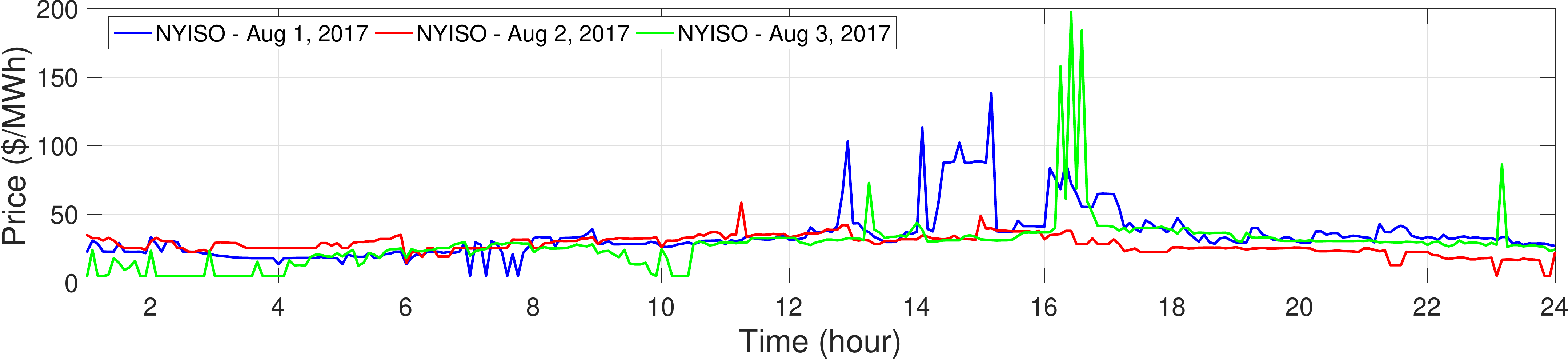}}\vspace{-3mm}
		\subfigure[German Electricity Market]{\label{fig:de_prices}
			\includegraphics[angle=0,scale=0.18]{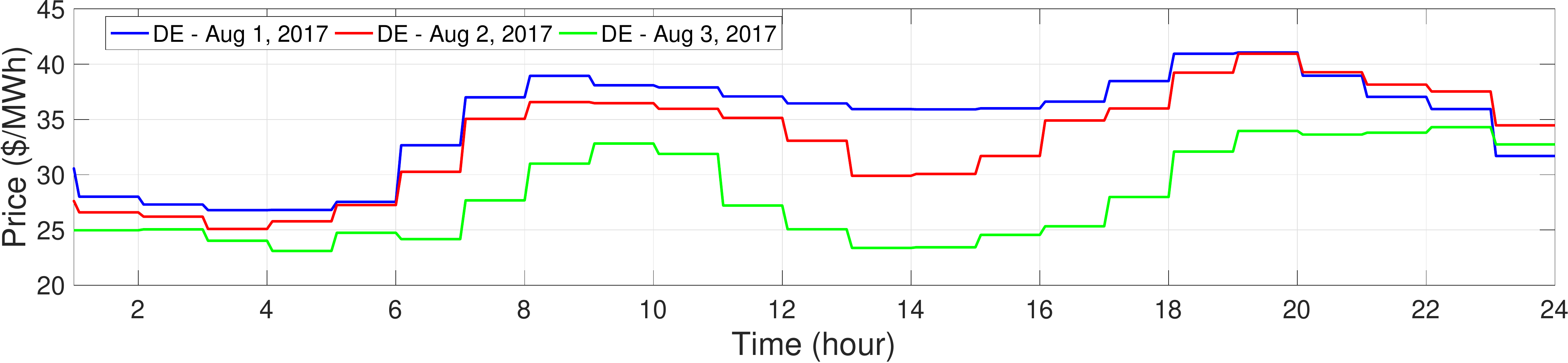}}\hspace{-1mm}
	\end{center}	
	\vspace{-6mm}
	\caption{The energy price dynamics in NYISO and German (DE) Market in three consecutive days in August 2017}
	\label{fig:market_prices}
\end{figure}

2) The data center energy demand is highly unpredictable because user demand for Internet services is extremely variable. 
Further, the sophisticated optimization algorithms used to improve the energy efficiency of data center's internal operations~\cite{DeepMind} can further increase the unpredictable variability of the energy demand. 


In addition, recently, several data centers are equipped with on-site renewable sources, e.g., Google data center in Belgium~\cite{Belgium}. The energy production level of renewable sources is uncertain and intermittent (exhibits high fluctuations)~\cite{liu2012renewable}.  For instance, energy production from  solar panels can change in a matter of  minutes due to cloud cover moving in to obscure the sun. This may lead to an increased uncertainty in net energy demand of data center. 
\begin{figure}[!h]
	\centering
	\begin{center}
		\subfigure[United States]{\label{fig:us_city}
			\includegraphics[angle=0,scale=0.18]{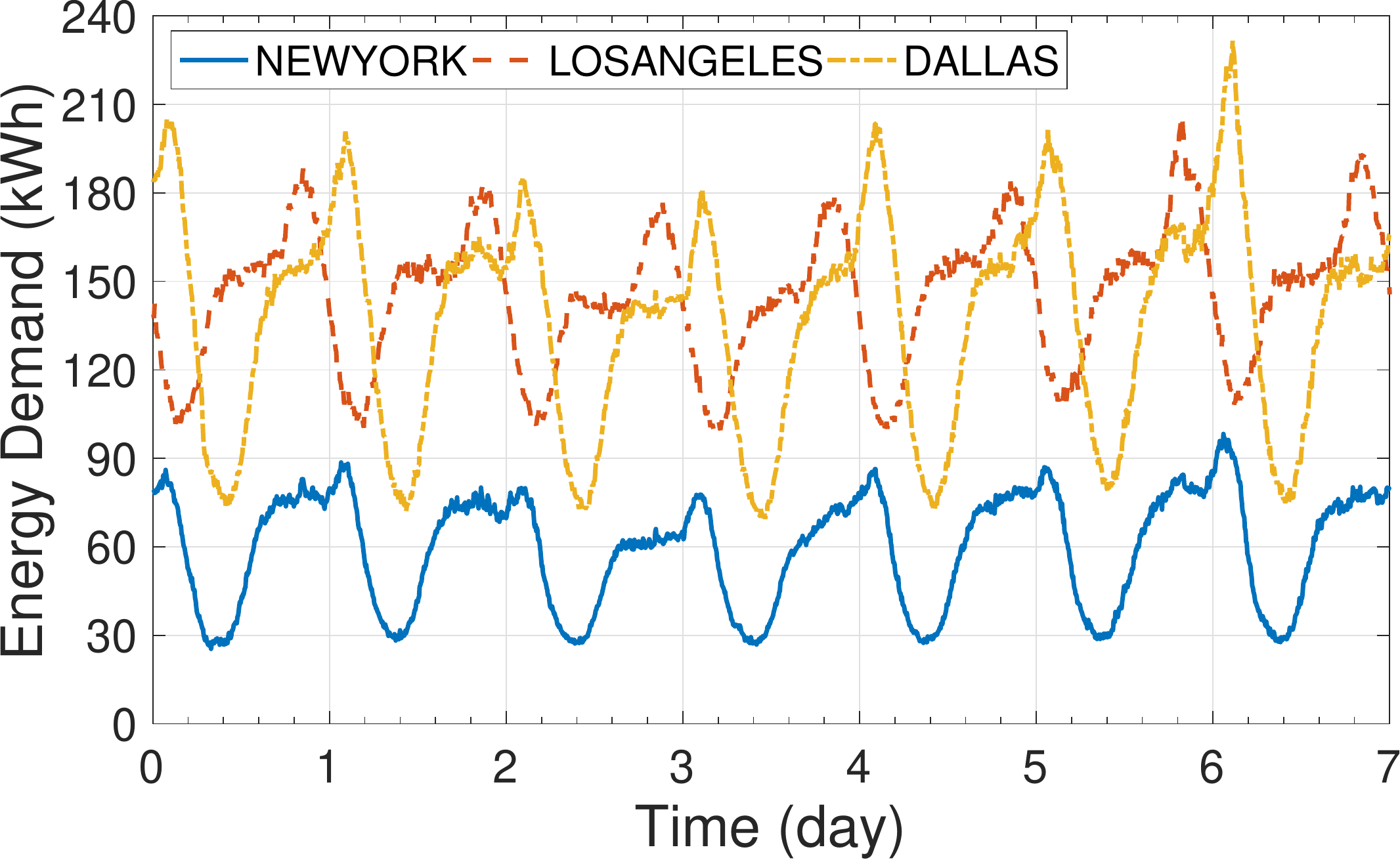}}\hspace{-1mm}
		\subfigure[Net energy demand with wind]{\label{fig:LOSANGELES_wind}
			\includegraphics[angle=0,scale=0.18]{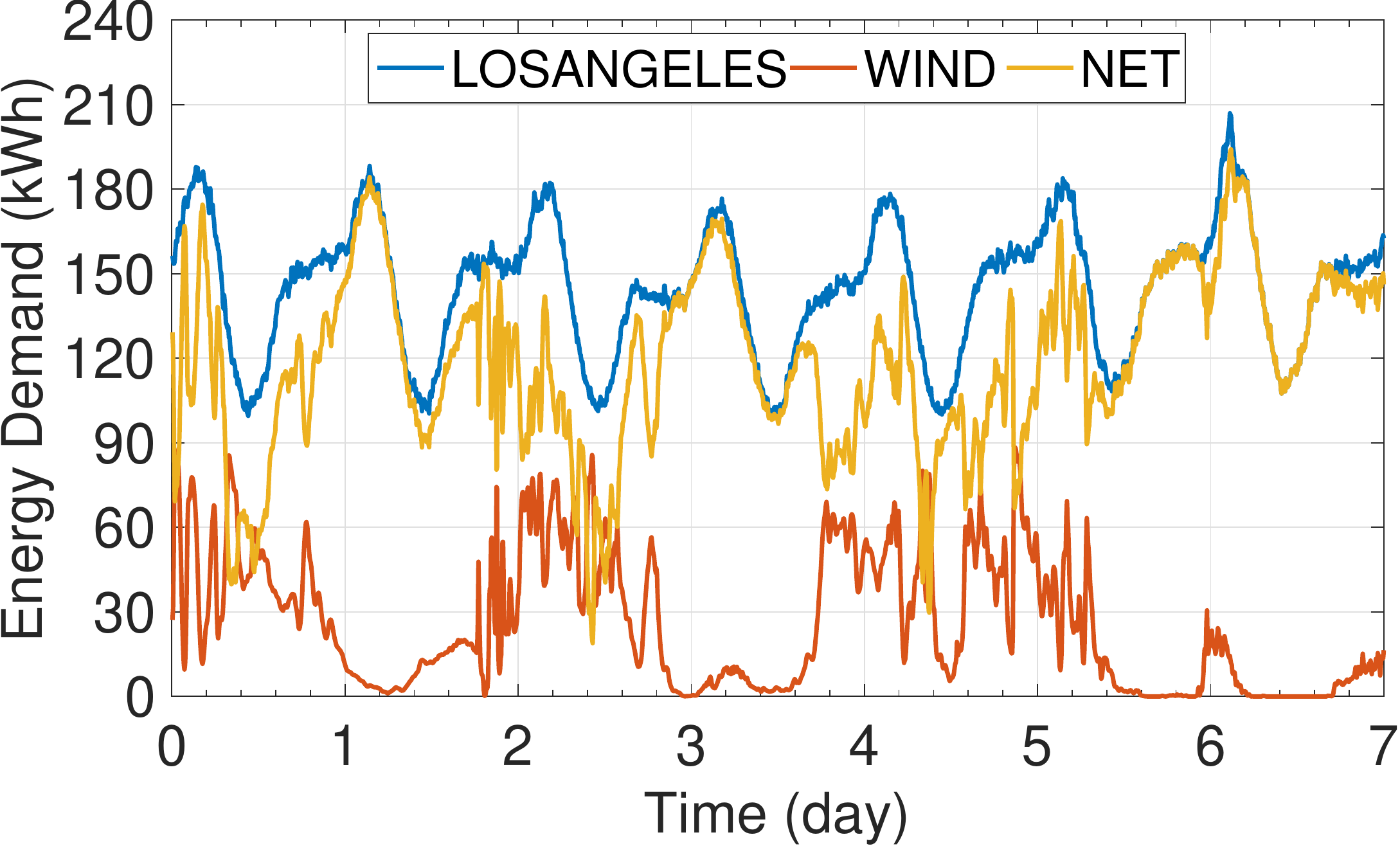}}\hspace{-1mm}
	\end{center}	
	\vspace{-5mm}
	\caption{Data center energy demand in different locations}
	\label{fig:dc_power_demand}
\end{figure}

To represent typical data center energy demand with and without renewable sources, we collected the energy consumption of Akamai's server clusters in different cities, some of which are depicted in Figure~\ref{fig:us_city} (details in Appendix~\ref{app:data}). The server clusters are part of  Akamai's global CDN~\cite{nygren2010akamai} that serves about a quarter of the Web traffic worldwide. We observe that the total energy consumption in different data centers has a relative regular (daily) pattern. However, by injecting renewable generation, the net demand, as depicted in Figure~\ref{fig:LOSANGELES_wind}, exhibits a large degree of variability.

\paragraph{Related Algorithmic Problems}
Note that \prob is related to several algorithmic problem including one-way trading~\cite{Yaniv01}, secretary problem~\cite{babaioff2008online}, $k$-search problem~\cite{lorenz2009optimal}, and online linear programming with covering constraints~\cite{buchbinder2009online}. It is related since in all above problems and \prob, the goal is to search for the optimum. In contrast, \prob is unique since it is the first that tackles inventory management constraints.  

Perhaps the most closely related problem to \prob is the minimization variant of $k$-search, known as $k$-min search problem~\cite{lorenz2009optimal}. In this problem, a player wants to buy $k \geq 1$ units of an asset with the goal of minimizing the cost. At any slot $t = \{1,\ldots,T\}$, the player is presented a price $p(t)$, and must immediately decide whether or not to buy \textit{some integer units} of the asset given $p(t)$. By setting $d(t)=0, t\in \{1,\dots,T-1\}$ and $d(T) = B$ and allowing \textit{fractional purchase}, \prob degenerates to the $k$-min search problem. Hence, \probeqs is an extension of continuous version of $k$-min search problem, or equivalently the minimization version of one-way trading problem since one-way trading problem can be viewed as the $k$-max search problem with $k\rightarrow\infty$. 

The other category is online linear programs with covering constraints~\cite{buchbinder2009online}. As compared to~\cite{buchbinder2009online}, \prob has a more specific category of covering constraints, however, the inventory management constraint in \prob couples the covering constraints across different time slots, which results in more challenging problem than basic online covering linear programs.

\section{Optimal Online Algorithms}
\label{sec:solution}

In \S\ref{sec:batman}, we propose \ocm, an online algorithm for a basic version of \prob without rate constraints. Then in~\S\ref{sec:analysis}, we prove that \ocm achieves the optimal competitive ratio. In \S\ref{sec:batmanrate} and based on the insights from design of \ocm, we propose \ocmrate that tackles the general \prob problem with rate constraints, and prove that it achieves the optimal competitive ratio.


\subsection{\ocm: An Online Algorithm }
\label{sec:batman}
Suppose that the inventory had only capacity constraints, but no rate constraints, i.e.,  $\rho_c = \rho_d = B$. In \prob, one can rewrite constraint~\eqref{eq:dt_cover} as $\xt \geq \dt - b(t-1)$, and remove constraint~\eqref{eq:dt_pack}.
In the next, we propose \ocm that finds a solution to \prob without rate constraints in online manner. 

\subsubsection{The Design of \ocm}
The high-level idea of designing \ocm is to store the asset when the market price is cheap and use the stored asset when the price is expensive. However, the dynamic pricing and dynamic demand make this decision making challenging. The main ideas of  \ocm are: (i) \textit{adaptive reservation based on storage utilization} that tackles the challenges due to price dynamics; and (ii) \textit{construction of virtual storages} that tackle the challenges due to demand dynamics.

\paragraph{Adaptive Reservation Price}
\ocm deals with the price dynamics by defining the notion of \textit{reservation price}. Having a properly constructed reservation price, \ocm stores the asset if the current price is cheaper than the reservation price; otherwise, it releases the asset from the storage.
\ocm adaptively determines the reservation price based on the available storage space. Intuitively, once the storage level is low, it is more eager to store asset, hence, it accepts higher prices. On the other hand, at high storage utilization, it stores the asset if the price is low. This is different from fixed reservation design introduced in~\cite{chau2016cost} that determines the reservation prices without considering the current storage utilization.


\paragraph{Constructing Virtual Storages}
\ocm deals with the demand dynamics by defining the notion of virtual storages.
\ocm views the demand in each time slot as an asset that must be purchased from the market with some degree of freedom obtained by shifting it using the storage. To utilize this opportunity, \ocm constructs several virtual storages to record the satisfied amount of the demand from the market. Specifically, in each slot with $\dt >0$, \ocm initiates a virtual storage whose capacity is equal to the demand, and its initial level is empty. \ocm also renews the virtual storage units once the actual storage becomes empty.

\paragraph{The Details of \ocm}
By summarizing the pseudocode of \ocm in Algorithm~\ref{alg:alg_1}, we proceed to discuss the details.
In all algorithms and analysis, we assume that the initial storage level is zero, i.e., $b(1) = 0$.
For notational convenience, we represent the physical storage as the first virtual storage and define $B_1=B$ as its capacity (Line~\ref{algline:set_cap}) and $v$ is the current number of virtual storages given positive demand that evolves over time (Line~\ref{algline:v_update} indicates creating a new virtual storage for the current slot and Line~\ref{algline:renew} renews the virtual storages). Let $B_{i}$ be the capacity of $i$-th virtual storage and $B_v = d(t)$, i.e., we set the value of new virtual storage to the current demand. Let $\xi_{i}$ be the reservation price associated to the virtual storage $i$ with initial value of $p_{\max}/\alpha$, where $\alpha >0$ is a parameter that will be carefully chosen based on the competitive analysis.

\begin{algorithm}[!t]
	\caption{The \ocm algorithm for each $t\in\mathcal{T}$}
	\label{alg:alg_1}
	\begin{algorithmic}[1]
		\State
        \textsc{// Initialization: at $t =1$}

		\State $B_1 \leftarrow B$; \textsc{// the capacity of physical storage} \label{algline:set_cap}
		\State $v \leftarrow 1$; \textsc{// the number of virtual and physical storages}
		\State $\xi_1\leftarrow [p_{\max}/\alpha]$; \textsc{// reservation price of physical storage}
		
		\textsc{// The main algorithm for $t$}

		\If {$d(t)>0$}
		\State $v\leftarrow v+1$\label{algline:v_update}
		\State $B_v \leftarrow d(t)$
		\State $\xi_v \leftarrow p_{\max}/\alpha$
		\EndIf
		\State $\xit\leftarrow [G_{B_i}(p(t))-G_{B_i}(\xi_{i})]^+, \quad \forall 1\leq i\leq v$ \label{algline:xi}
		\State $\xi_i \leftarrow \min\{\pt,\xi_i\}, \quad \forall 1\leq i\leq v$ \label{algline:xii}
		\State $\xt\leftarrow \max\left\{\sum_{i=1}^{v}\xit,[d(t)-b(t-1)]^{+}\right\}$ \label{algline:xt}
		\State $\bt\leftarrow b(t-1)+\xt-d(t)$ \label{algline:bt}

        \textsc{// Renew virtual storages} \label{algline:renew}

		\State \textbf{if} {$\bt=0$} \textbf{then} $v \leftarrow 1; \xi_1\leftarrow p_{\max}/\alpha$;
	\end{algorithmic}
\end{algorithm}

The cornerstone of \ocm is in Line~\ref{algline:xi} where the procurement amount for each virtual storage, i.e., $x_i(t)$, is set.
\ocm defines $G_{B_i}(\xi_i)$ as the reservation function to determine the amount of asset to be stored in the $i$-th virtual storage in each slot.
This function represents the target amount of stored assets in $i$-th virtual storage when the reservation price is $\xi_i$.
In slot $t$, \ocm stores additional amount of $G_{B_i}(p(t))-G_{B_i}(\xi_i)$ into virtual storage $i$ if the current price, $p(t)$, is less than $\xi_i$; otherwise, it stores nothing. Both situations can be stated in the following compact form
\begin{equation}
\label{eq:xit}
x_i(t)=\left[G_{B_i}(p(t))-G_{B_i}(\xi_i)\right]^+.
\end{equation}
In this way, \ocm logically allocates $x_{i}(t)$ units of asset to storage $i$ (Line~\ref{algline:xi}), and the reservation price will be updated to $\min\{\xi_{i},p(t)\}$ (Line~\ref{algline:xii}). In other words, $\xi_i$ records the minimum seen market price during the lifetime of $i$-th virtual storage.

The main contribution of \ocm is the design function $G_{B_i}(p)$ such that it achieves the optimal competitive ratio. To accomplish this, we choose the following function:
\begin{equation}
G_{B_i}(p)=\alpha B_i\ln \left[\left(1-\frac{p}{p_{\max}}\right)\frac{\alpha}{\alpha-1}\right],~p\in\left[p_{\min},\frac{p_{\max}}{\alpha}\right],
\end{equation}
where
\begin{equation}
\label{eq:alpha2}
\alpha = \left(W\left(-\frac{\theta -1}{\theta\exp(1)}\right)+1\right)^{-1},
\end{equation}
and $W$ denotes \textit{Lembert-W function} defined as inverse of $f(z) = z\exp(z)$, and $\theta = p_{\max}/ p_{\min}$ is the  price fluctuation ratio.
Figure~\ref{fig:procurement_amount} depicts function $G_{B_i}(p) \in [0, B_i],$ and $p \in [p_{\min}, p_{\max}/\alpha]$ as a decreasing function. It also demonstrates how to determine the reservation amount for virtual storage $i$. Further, when the price $p$ is larger than and equal to $p_{\max}/\alpha$, $G_{B_i}(p)=0$. When the price is equal to $p_{\min}$, the reservation amount is equal to $\alpha B_i\ln \left[\left(1-\frac{1}{\theta}\right)\frac{\alpha}{\alpha-1}\right]$. By substituting the value of $\alpha$ from Equation~\eqref{eq:alpha2}, we have $G_{B_i}(p_{\min})=B_i$, which means when the price is minimum, \ocm stores the full capacity of the storage.

The last step is to determine the aggregate procurement quantity $x(t)$ (Line~\ref{algline:xt}). To satisfy the demand constraint, i.e., $x(t) \geq d(t) - b(t-1)$, we calculate $x(t)$ as follows
\begin{equation}
\label{eq:xt}
x(t) = \max\left\{\sum_{i=1}^{v}x_{i}(t),d(t)-b(t-1)\right\}.
\end{equation}
When $\sum_{i=1}^{v}x_{i}(t)<d(t)-b(t)$, \ocm discharges the physical storage completely to meet the demand, hence, we have $b(t)=0$. In this case, we renew all virtual storages to the initial state of having only the physical storage (Line~\ref{algline:renew}). The intuition is that with fully discharging the physical storage, we exhausted the capability of shifting the demand using the storage, hence we renew the process.


\begin{figure}[!t]
	\centering
	\subfigure[An illustration of function $G_{B_i}(p)$ and determining $x_i(t)$ as the procurement amount for virtual storage $i$]{\label{fig:procurement_amount}
		\includegraphics[width=0.23\textwidth]{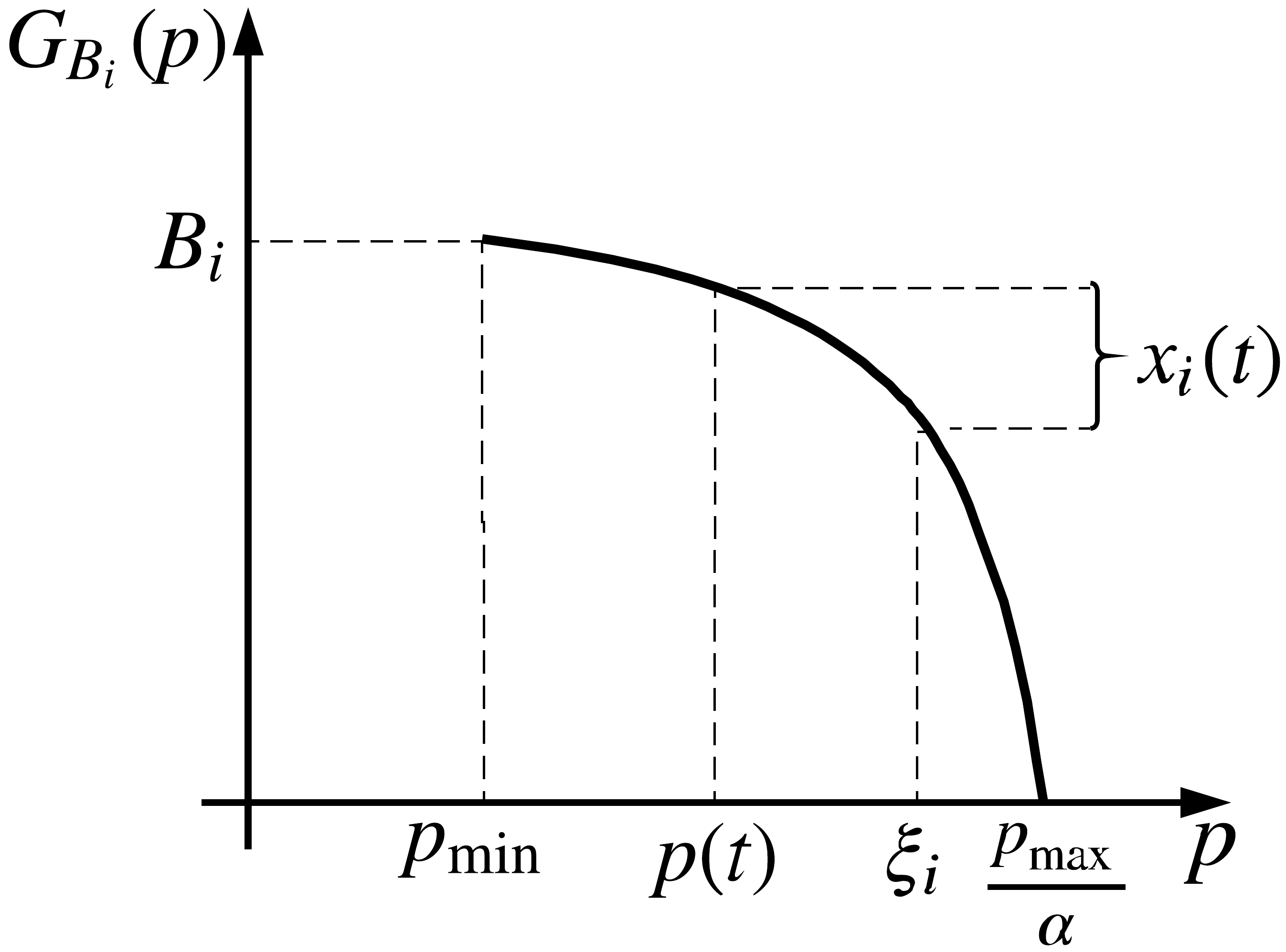}}\hspace{1mm}
	\subfigure[An illustration of competitive ratio in Equation~\eqref{eq:alpha} as a function of $\theta$  in comparison with $\sqrt{\theta}$~\cite{chau2016cost} and $\log \theta$.]{\label{fig:cr_theta}
		\includegraphics[width=0.22\textwidth]{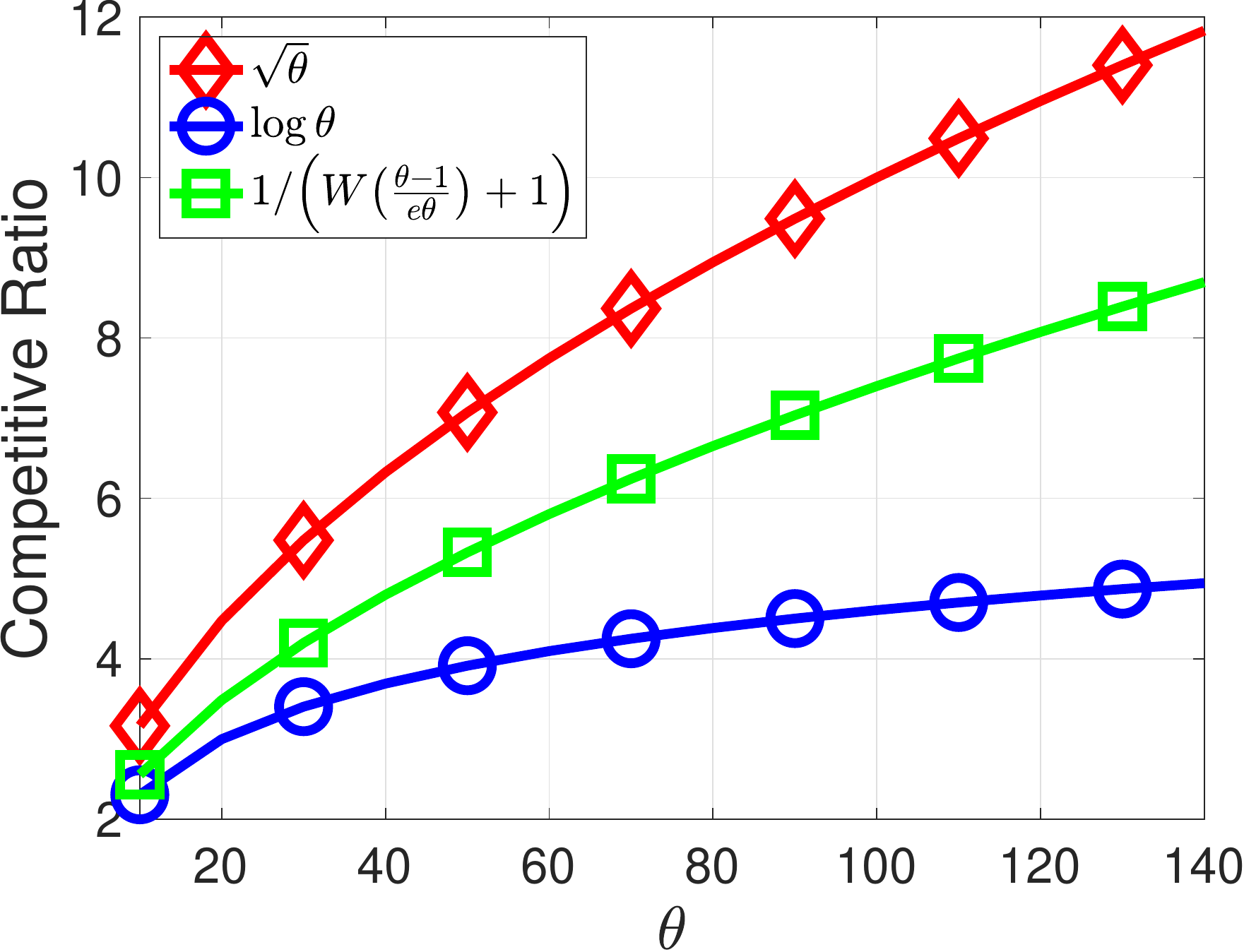}}
	\caption{Function $G_{B_i}(p)$ and growth of competitive ratio}
	\label{fig:cr_fun}
\end{figure}

\begin{theorem}
	\label{thm:feasibility}
	\ocm generates a feasible solution to \emph{\probeqs}.
\end{theorem}
The proof for covering the demand is straightforward because of Equation~\eqref{eq:xt}. The range of function $G_{B_i}(p)$ along with the renewal process in Line~\ref{algline:renew} guarantees the respecting the capacity constraint of physical storage. The proof is formally given in  Appendix~\ref{app:thm3.1}.

\subsection{Competitive Analysis of \ocm}
\label{sec:analysis}
This section proves the result in Theorem~\ref{thm:cr}.
Before the formal proof, we state the following remark about the competitive ratio of \ocm.
In Figure~\ref{fig:cr_theta}, we depict the value of the optimal competitive ratio $\alpha$ as a function of price fluctuation ratio $\theta$, as compared to a sub-optimal competitive ratio of $\sqrt{\theta}$ achieved in~\cite{chau2016cost} in a slightly different setting, and $\log \theta$ as a baseline, and it shows that its growth is less than the $\sqrt{\theta}$. However, series expansion of Lambert-$W$ function shows that $\alpha$ as indicated in~\eqref{eq:alpha} is in order of $\Theta(\sqrt \theta)$. This is in contrast to the best competitive ratio for equivalent maximization problems, e.g., $k$-max search that achieves competitive ratio of like $\log \theta$~\cite{Yang17,Yaniv01,yang2017optimal}. This shows that minimization and maximization version of search problems behave differently in terms of best possible competitive algorithms.  

In what follows, we prove the result in Theorem~\ref{thm:cr}. First, we give the preliminaries (\S\ref{sec:def}). Second, we characterize an upper bound on the cost of \ocm (Lemma~\ref{lem:online_cost}). Third, a lower bound on the offline optimum is obtained (Lemma~\ref{lem:offline_cost}). Forth, we prove the competitive ratio by comparing these two values. Finally, we prove the optimality of the competitive ratio in \S\ref{sec:opt_cr}.

\subsubsection{Definitions and Preliminaries}
\label{sec:def}
First, to be consistent to the notations in this paper, ideally we must denote all the inputs and variables with index $t$. For notation brevity, however, we slightly abuse the notations by dropping index $t$ in the analysis. 

\begin{myDef}
Define $\omega\in \Omega$ as an input instance to \emph{\probeqs} including the price and demand, i.e.,
\begin{equation*}
\omega\overset{\textit{\text{def}}}{=}\left[\omega(t)=\langle p(t),d(t)\rangle\right]_{t\in \mathcal{T}}.
\end{equation*}
Moreover, $\emph{\textsf{cost}}_{\omega}(\ocm)$ is the cost of \ocm under instance $\omega$ and $\emph{\textsf{cost}}_{\omega}(\ofa)$ is the offline optimal cost under $\omega$. We drop subscript $\omega$ from the costs when we are not focusing on a particular $\omega$.
\end{myDef}

\begin{myDef}
	\ocm is $\alpha$-competitive, if for any $\omega\in \Omega$
	\begin{equation}
	\label{eq:cr_def}
	\emph{\textsf{cost}}_{\omega}(\ocm)\leq  \alpha \cdot \emph{\textsf{cost}}_{\omega}(\ofa) + \emph{\textsf{cons}},
	\end{equation}
	where $\emph{\textsf{cons}} \geq 0$ is a constant number.
\end{myDef}

In proofs, we set the value of ``$\textsf{cons}$'' to capture the special input scenarios to \prob. 

\begin{myDef}
	\label{def:res_period}
	Reservation and idle periods. The time horizon $T$ can be divided into two type of periods: the reservation period, which contains the interval between beginning to charge and fully discharging of the storage; and the idle period, which corresponds to the interval that lies between two adjacent reservation periods.
\end{myDef}

The following is the list of additional notations that we use in the analysis.
Consider an instance in which the executing of \ocm results in $n$ reservation periods totally.
During the $i$-th reservation period, $i\leq n$, we assume there are totally $\hat{v}_i$ virtual storage units created. Let $\hat{\xi}_{i,j}$ be the final reservation price of the $j$-th virtual storage during the $i$-th reservation period.
Let $B_{i,j}$ be the capacity of the $j$-th storage and $\hat{b}_{i,j}$ be the final storage level correspondingly.
Obviously, we have $\hat{b}_{i,j}=G_{B_{i,j}}(\hat{\xi}_{i,j})$. Let $D$ be the total demand during the time horizon, i.e., $D = \sum_{t\in\mathcal{T}} d(t)$, and $\hat{b}$ be the final storage level of the physical storage, i.e., $\hat{b} = b(T)$. In addition, let $F_i(\beta)$ be the minimum cost of purchasing $\beta$ units of asset during $i$-th reservation period, and $\tilde{p}$ be the minimum price during the idle periods. Finally, the inverse of function $G_{B_i}$ is defined as
 \begin{equation}
 \label{eq:g_inverse}
 G^{-1}_{B_i}(b)=p_{\max}\left[1-\left(1-\frac{1}{\alpha}\right)\exp\left(\frac{b}{\alpha B_i}\right)\right],~b\in[0,B_i].
 \end{equation}

\subsubsection{The Proof of Theorem~\ref{thm:cr}}
The following lemma characterizes an upper bound on the cost of \ocm, during the time horizon. The proof of all lemmas in this section are given in Appendix~\ref{app:proof_batman}.

\begin{lemma}
	\label{lem:online_cost}
	\emph{\costocm} is upper bounded by
	\begin{equation*}
		\emph{\costocm}\leq\sum_{i=1}^{n}\sum_{j=1}^{\hat{v}_i}\int_{0}^{\hat{b}_{i,j}}G^{-1}_{B_{i,j}}(b)db+\left(D-\sum_{i=1}^{n}\sum_{j=1}^{\hat{v}_i}\hat{b}_{i,j}+\hat{b}\right) p_{\max}.
	\end{equation*}
\end{lemma}

The next lemma yields a lower bound on the cost of offline optimal solution.
\begin{lemma}
\label{lem:offline_cost}
\emph{\costopt} is lower bounded by
	\begin{equation*}
\emph{\costopt} \geq \sum_{i=1}^{n}F_i(\beta_i)+\left(D-\sum_{i=1}^{n} \beta_i\right) \tilde{p}.
\end{equation*}
\end{lemma}


Now, we proceed to prove the competitive ratio. First, we consider the simple case, where $D=0$. In this trivial case, $\costopt=0$, that of \ocm is at most $Bp_{\max}$. Obviously, \ocm is $\alpha$-competitive since it satisfies the definition of competitive ratio in Equation~\eqref{eq:cr_def} by setting $\textsf{cons} = Bp_{\max}$.

Second, we focus on a realistic case, in which $D>0$. If the minimum price during the time horizon is larger than or equal to $p_{\max}/\alpha$, the cost of \ocm is at most $p_{\max}(B+D)$ and that of \ofa is lower bounded by $\frac{p_{\max}}{\alpha}D$. It is easy to see that \ocm is $\alpha$-competitive according to the definition.
We only consider the case where the minimum price during the time horizon is less than $p_{\max}/\alpha$, and obviously, the minimum price occurs in the reservation period. We have $\sum_{i=1}^{n}F_i(\beta_i)>0$. Using the results in lemmas~\ref{lem:online_cost} and~\ref{lem:offline_cost} we have
\begin{equation}
\label{eq:ratio}
	\begin{split}		
	& \frac{\costocm-\hat{b}p_{\max}}{\costopt} \\
		\leq &\frac{\sum\limits_{i=1}^{n}\sum\limits_{j=1}^{\hat{v}_i}\int_{0}^{\hat{b}_{i,j}}G^{-1}_{B_{i,j}}(b)db+\left(D-\sum\limits_{i=1}^{n}\sum\limits_{j=1}^{\hat{v}_i}\hat{b}_{i,j}\right)\cdot p_{\max}}{\sum\limits_{i=1}^{n}F_i(\beta_i)+\left(D-\sum\limits_{i=1}^{n} \beta_i\right)\cdot \tilde{p}} \\
=&\frac{\mathcal{Q}+\left(\sum\limits_{i=1}^{n} \beta_i-\sum\limits_{i=1}^{n}\sum\limits_{j=1}^{\hat{v}_i}\hat{b}_{i,j}\right) p_{\max}+\left(D-\sum\limits_{i=1}^{n} \beta_i\right)\cdot p_{\max}}{\sum\limits_{i=1}^{n}F_i(\beta_i)+\left(D-\sum\limits_{i=1}^{n} \beta_i\right)\cdot \tilde{p}} \\
\leq & \max\left\{\frac{\mathcal{Q}+\left(\sum\limits_{i=1}^{n} \beta_i-\sum\limits_{i=1}^{n}\sum\limits_{j=1}^{\hat{v}_i}\hat{b}_{i,j}\right) p_{\max}}{\sum\limits_{i=1}^{n}F_i(\beta_i)},\frac{\left(D-\sum\limits_{i=1}^{n} \beta_i\right)\cdot p_{\max}}{\left(D-\sum\limits_{i=1}^{n} \beta_i\right)\cdot \tilde{p}} \right\}\\
		\leq &\max\left\{\frac{\mathcal{Q}+\left(\sum\limits_{i=1}^{n} \beta_i-\sum\limits_{i=1}^{n}\sum\limits_{j=1}^{\hat{v}_i}\hat{b}_{i,j}\right) \cdot p_{\max}}{\sum\limits_{i=1}^{n}F_i(\beta_i)},\alpha \right\}. 
	\end{split}
\end{equation}
where $$\mathcal{Q}=\sum\limits_{i=1}^{n}\sum\limits_{j=1}^{\hat{v}_i}\int_{0}^{\hat{b}_{i,j}}G^{-1}_{B_{i,j}}(b)db.$$
In Equation (\ref{eq:ratio}), the first inequality is by Lemma \ref{lem:online_cost} and \ref{lem:offline_cost}, the second inequality is by $D-\sum\limits_{i=1}^{n} \beta_i\geq 0$  and the third inequality is by $p_{\max}/\tilde{p}\leq\alpha$.

The following two lemmas provide an upper bound for Equation~(\ref{eq:ratio}), as a main step to prove the competitive ratio.

\begin{lemma}
	\label{lem:a_critical_property_on_g}
	Given $G^{-1}_{B_{i,j}}(\hat{b}_{i,j}), i\in\{0,1,\ldots,n\}$ in Equation~\eqref{eq:g_inverse}, we have
	\begin{equation}
	\label{eq:lemma3}
	\frac{\int_{0}^{\hat{b}_{i,j}}G^{-1}_{B_{i,j}}(b)db+(B_{i,j}-\hat{b}_{i,j})p_{\max}}{\hat{\xi}_{i,j}B_{i,j}}= \alpha,~\forall~\hat{b}_{i,j}\in[0,B_{i,j}].
	\end{equation}
\end{lemma}
\begin{proof}
	By substituting Equation~\eqref{eq:g_inverse}, we first calculate the second term in numerator of Equation~\eqref{eq:lemma3} as follows
	\begin{equation*}
	\begin{split}
	&\int_{b=0}^{\hat{b}_{i,j}}G^{-1}_{B_{i,j}}(b)db \\
	=&p_{\max}\!\left[b-\left(1-\frac{1}{\alpha}\!\right)\exp\left(\frac{b}{\alpha B_{i,j}}\right)\alpha B_{i,j}\right]\Bigg|_{b=0}^{\hat{b}_{i,j}} \\
	=&p_{\max}\left[\hat{b}_{i,j}-\left(1-\frac{1}{\alpha}\right)\exp\left(\frac{\hat{b}_{i,j}}{\alpha B_{i,j}}\right)\alpha B_{i,j}\right] \\
	&+p_{\max}\left(1-\frac{1}{\alpha}\right)\alpha B_{i,j}. \\
	\end{split}
	\end{equation*}
	Then, we calculate the numerator
	\begin{eqnarray}
	&&(B_{i,j}-\hat{b}_{i,j})p_{\max}+\int_{b=0}^{\hat{b}_{i,j}}G^{-1}_{B_{i,j}}(b)db \nonumber\\
	&=&\alpha B_{i,j} \left(p_{\max}\left[1-\left(1-\frac{1}{\alpha}\right)\exp\left(\frac{\hat{b}_{i,j}}{\alpha B_{i,j}}\right)\right]\right)\nonumber\\
	&=&\alpha B_{i,j} G^{-1}_{B_{i,j}}(\hat{b}_{i,j}).\label{eq:num}
	\end{eqnarray}
	Substituting~\eqref{eq:num} into~\eqref{eq:lemma3} completes the proof.
\end{proof}

Using the result in Lemma~\ref{lem:a_critical_property_on_g}, we have the following result.
\begin{lemma}
	\label{lem:neq}
\begin{equation*}
	\frac{\sum\limits_{i=1}^{n}\sum\limits_{j=1}^{\hat{v}_i}\int_{0}^{\hat{b}_{i,j}}G^{-1}_{B_{i,j}}(b)db+\left(\sum\limits_{i=1}^{n} \beta_i-\sum\limits_{i=1}^{n}\sum\limits_{j=1}^{\hat{v}_i}\hat{b}_{i,j}\right) p_{\max}}{\sum\limits_{i=1}^{n}F_i(\beta_i)}\leq\alpha.
\end{equation*}
\end{lemma}
\begin{proof}
We prove the result in Lemma~\ref{lem:neq} by contradiction. Assume
\begin{equation*}
\frac{\sum\limits_{i=1}^{n}\sum\limits_{j=1}^{\hat{v}_i}\int_{0}^{\hat{b}_{i,j}}G^{-1}_{B_{i,j}}(b)db+\left(\sum\limits_{i=1}^{n} \beta_i-\sum\limits_{i=1}^{n}\sum\limits_{j=1}^{\hat{v}_i}\hat{b}_{i,j}\right) p_{\max}}{\sum\limits_{i=1}^{n}F_i(\beta_i)}>\alpha.
\end{equation*}
By statement (3) in Lemma \ref{lem:opt_cost}, we have
\begin{equation*}
\begin{split}
&\frac{\mathcal{Q}+\left(\sum\limits_{i=1}^{n}\sum\limits_{j=1}^{\hat{v}_i}B_{i,j}-\sum\limits_{i=1}^{n}\sum\limits_{j=1}^{\hat{v}_i}\hat{b}_{i,j}\right) p_{\max}}{\sum\limits_{i=1}^{n}F_i\left(\sum\limits_{j=1}^{\hat{v}_i}B_{i,j}\right)} \\
=&\frac{\mathcal{Q}+\left(\sum\limits_{i=1}^{n}\sum\limits_{j=1}^{\hat{v}_i}\beta_{i,j}-\sum\limits_{i=1}^{n}\sum\limits_{j=1}^{\hat{v}_i}\hat{b}_{i,j}\right) p_{\max}+\sum\limits_{i=1}^{n}\sum\limits_{j=1}^{\hat{v}_i}(B_{i,j}-\beta_{i,j})p_{\max}}{\sum\limits_{i=1}^{n}F_i\left(\sum\limits_{j=1}^{\hat{v}_i}\beta_{i,j}\right)+\sum\limits_{i=1}^{n}\left[F_i\left(\sum\limits_{j=1}^{\hat{v}_i}B_{i,j}\right)-F_i\left(\sum\limits_{j=1}^{\hat{v}_i}\beta_{i,j}\right)\right]}\\
\geq & \frac{\mathcal{Q}+\left(\sum\limits_{i=1}^{n}\sum\limits_{j=1}^{\hat{v}_i}\beta_{i,j}-\sum\limits_{i=1}^{n}\sum\limits_{j=1}^{\hat{v}_i}\hat{b}_{i,j}\right) p_{\max}+\sum\limits_{i=1}^{n}\sum\limits_{j=1}^{\hat{v}_i}(B_{i,j}-\beta_{i,j})p_{\max}}{\sum\limits_{i=1}^{n}F_i\left(\sum\limits_{j=1}^{\hat{v}_i}\beta_{i,j}\right)+\sum\limits_{i=1}^{n}\sum\limits_{j=1}^{\hat{v}_i}(B_{i,j}-\beta_{i,j})\frac{p_{\max}}{\alpha}}\\
>&\alpha.
\end{split}
\end{equation*}

During the lifetime of the $j$-th virtual storage of the $i$-th reservation period,
the minimum electricity price is $\hat{\xi}_{i,j}$. 
When the procurement amount during the $i$-th reservation period is $\sum_{j=1}^{\hat{v}_i}B_{i,j}$,
the cost of the optimal algorithm satisfies
\begin{equation*}
\sum\limits_{i=1}^{n}F_i\left(\sum_{j=1}^{\hat{v}_i}B_{i,j}\right)\geq \sum\limits_{i=1}^{n}\sum_{j=1}^{\hat{v}_i}\hat{\xi}_{i,j}B_{i,j}.
\end{equation*}

Thus, we have
\begin{equation*}
\begin{split}
&\frac{\sum\limits_{i=1}^{n}\sum\limits_{j=1}^{\hat{v}_i}\int_{0}^{\hat{b}_{i,j}}G^{-1}_{B_{i,j}}(b)db+\left(\sum\limits_{i=1}^{n}\sum\limits_{j=1}^{\hat{v}_i}B_{i,j}-\sum\limits_{i=1}^{n}\sum\limits_{j=1}^{\hat{v}_i}\hat{b}_{i,j}\right) p_{\max}}{\sum\limits_{i=1}^{n}\sum\limits_{j=1}^{\hat{v}_i}\hat{\xi}_{i,j}B_{i,j}} \\
=& \frac{\sum\limits_{i=1}^{n}\sum\limits_{j=1}^{\hat{v}_i}\left[\int_{0}^{\hat{b}_{i,j}}G^{-1}_{B_{i,j}}(b)db+\left(B_{i,j}-\hat{b}_{i,j}\right) p_{\max}\right]}{\sum\limits_{i=1}^{n}\sum\limits_{j=1}^{\hat{v}_i}\hat{\xi}_{i,j}B_{i,j}}>\alpha.
\end{split}
\end{equation*}

That means that there is at least a pair of $i$ and $j$ such that
\begin{equation*}
\frac{\int_{0}^{\hat{b}_{i,j}}G^{-1}_{B_{i,j}}(b)db+(B_{i,j}-\hat{b}_{i,j})p_{\max}}{\hat{\xi}_{i,j}B_{i,j}}> \alpha.
\end{equation*}
The above equation contradicts the results in Lemma \ref{lem:a_critical_property_on_g}. This completes the proof.
\end{proof}

Using the result in Lemma~\ref{lem:neq} and Equation (\ref{eq:ratio}), we have
\begin{equation*}
\frac{\sum\limits_{i=1}^{n}\sum\limits_{j=1}^{\hat{v}_i}\int_{0}^{\hat{b}_{i,j}}G^{-1}_{B_{i,j}}(b)db+\left(D-\sum\limits_{i=1}^{n}\sum\limits_{j=1}^{\hat{v}_i}\hat{b}_{i,j}\right) p_{\max}}{\sum\limits_{i=1}^{n}F_i(\beta_i)+\left(D-\sum\limits_{i=1}^{n} \beta_i\right)\cdot \tilde{p}}\leq \alpha.
\end{equation*}

Thus, we have
\begin{equation*}
\costocm\leq \alpha \cdot \costopt + \hat{b} p_{\max}\leq \alpha \cdot \costopt + B p_{\max},
\end{equation*}
where $B p_{\max}$ is a constant. And the proof of Theorem \ref{thm:cr} is complete.

\subsubsection{The Optimality of the Competitive Ratio}
\label{sec:opt_cr} 
By setting  $d(t)=0, t\in \{1,\dots,T-1\}$ and $d(T) = B$, \probeqs degenerates to the $k$-min search problem, whose optimal competitive ratio~\cite[Theorem~2]{lorenz2009optimal} is exactly equal to that of \ocm. This directly results that $\alpha$ is a lower bound for the \probeqs as a generalized $k$-min search problem.


\subsection{\ocmrate: An Online Algorithm with Rate Constraints}
\label{sec:batmanrate}

\begin{algorithm}[!t]
\caption{The \ocmrate algorithm for each $t\in\mathcal{T}$}
\label{alg:alg_2}
\begin{algorithmic}[1]
\State \textsc{// Initialization: just at first slot}
\State $B_1 \leftarrow B$; $v \leftarrow 1$; $\xi_1\leftarrow p_{\max}/\alpha$
\vspace{0.0cm}

\textsc{// The main procedure for $t$}

\textsc{// Initialize a new virtual storage}
\If {$d(t)>0$}
\State $v\leftarrow v+1$
\State $\xi_{v}\leftarrow p_{\max}/\alpha$
\State $B_{v} \leftarrow$ \texttt{InitVS} ($p(t),d(t),\rho_d, v,\{\xi_i,B_i\}_{i=1:v-1},\varepsilon_1$)\label{algline:call_calculateBX}
\EndIf
\vspace{0.0cm}

\textsc{// Calculate the initial value for procurement amount}
\State $\hat{x}(t)\leftarrow\sum_{i=1}^{v}\left[G_{B_i}(p(t))-G_{B_i}(\xi_i)\right]^+$
\State $\xt \leftarrow \hat{x}(t)$ \label{algline:xt2}
\State $p \leftarrow p(t)$ \label{algline:xi2}
\vspace{0.0cm}

\textsc{// Check active output rate constraint}

\If {$\hat{x}(t)< [d(t)-\min\{b(t-1),\rho_d\}]^{+}$} \label{algline:discharge}
\State $\xt\leftarrow[d(t)-\min\{b(t-1),\rho_d\}]^{+}$
\EndIf
\vspace{0.0cm}

\textsc{// Check active input rate constraint}
\If {$\hat{x}(t)> \rho_c +d(t)$}  \label{algline:charge}
\State $\xt\leftarrow\rho_c +d(t)$
\State $p \leftarrow$ \texttt{CalRP} ($d(t),\rho_c, v, \{\xi_i,B_i\}_{i=1:v},\varepsilon_2$)
\EndIf
\vspace{0.0cm}

\State $\bt\leftarrow b(t-1)+\xt-d(t)$ \label{algline:bt}
\vspace{0.0cm}

 \label{algline:call_calculateRP}
\State $\xi_i \leftarrow \min\{p,\xi_i\}, \quad \forall 1\leq i\leq v$ \label{algline:xi3}
\vspace{0.0cm}


\State \textbf{if} $b(t)=0$ \textbf{then} $v \leftarrow 1$ and $\xi_1 \leftarrow p_{\max}/\alpha$

\end{algorithmic}
\end{algorithm}


In this section, we design \ocmrate that adds input and output rate constraints to \ocm. 
In Algorithm~\ref{alg:alg_2}, the pseudocode of \ocmrate is summarized. While the general flow is similar to that of \ocm, \ocmrate has two major extensions:

First, \ocmrate intelligently sets the capacity of virtual storages to respect output rate constraints.
The high-level intuition to set the capacity of virtual storage is that output (discharge) rate constraint limits the capability of using the storage in each slot, and hence it may not be possible to fully satisfy the demand by discharging the storage, so creating a virtual storage with capacity equal to the demand does not make sense.
This is done by calling sub-procedure \texttt{InitVS} in Line~\ref{algline:call_calculateBX} of \ocmrate, with details explained in \S\ref{sec:batmanrate_bx}.

Second, \ocmrate intelligently sets the value of reservation prices to respect the input (charge) rate constraints.
The high-level intuition is that the input rate constraint limits the amount of stored asset in each slot.
Hence, if the price is very cheap, \ocm might propose to store some amount that is beyond the capability of storage to input.
Hence, \ocmrate updates the reservation price based on the capability to store, i.e., input rate constraint.
This is done by calling the sub-procedure \texttt{CalRP} in Line~\ref{algline:call_calculateRP} of \ocmrate, with details presented in \S\ref{sec:batmanrate_rp}.

\subsubsection{Initializing the New Virtual Storage}
\label{sec:batmanrate_bx}

\begin{algorithm}[!t]
\caption{\texttt{InitVS} ($p(t),d(t),\rho_d,v,\{\xi_i,B_i\}_{i=1:v-1},\varepsilon_1$)}
\label{alg:alg_3}
\begin{algorithmic}[1]
\State $B_{v}\leftarrow 0,~B'_{v}\leftarrow d(t)$
\While{$|B'_{v}-B_{v}|>\varepsilon_1$}
\State $B_{v}\leftarrow B'_{v}$
\State $\hat{x}(t)\leftarrow\sum_{i=1}^{v}\left[G_{B_i}(p(t))-G_{B_i}(\xi_i)\right]^+$\label{algline:xhat}
\State $B'_{v}\leftarrow d(t)-[d(t)-\rho_d-\hat{x}(t)]^+$\label{algline:update_cap}
\EndWhile
\State \textbf{Output}: $B_{v}$
\end{algorithmic}
\end{algorithm}

At $t$, the \textit{preferred procurement amount}, denoted as $\hat{x}(t)$, the amount without considering output rate, should be calculated as
\begin{equation}
\label{eq:1}
\hat{x}(t)=\sum_{i=1}^{v}\left[G_{B_i}(p(t))-G_{B_i}(\xi_i)\right]^+,
\end{equation}
where $B_{v}$ is the new capacity of virtual storage.
Different from the \ocm that sets the capacity to $d(t)$, \ocmrate subtracts $[d(t)-\rho_d-\hat{x}(t)]^+$, i.e., the additional amount due to output rate constraint, from the capacity, hence
\begin{equation}
\label{eq:2}
B_{v}=d(t)-[d(t)-\rho_d-\hat{x}(t)]^+.
\end{equation}
Equations~\eqref{eq:1} and~\eqref{eq:2} show that $B_{v}$ and $\hat{x}(t)$ are dependent on each other. To address this, we devise the sub-procedure \texttt{InitVS} (Algorithm \ref{alg:alg_3}) to calculate the capacity of the new virtual storage.
\texttt{InitVS} captures this dependency and updates determining $B_{v}$ and $x(t)$ using a simple search algorithm in iterative manner with parameter $\varepsilon_1$ as the stopping criteria ($\varepsilon_1$ can be arbitrarily close to 0).
In Appendix~\ref{app:conv}, we prove that \texttt{InitVS} converges to a solution to Equations (\ref{eq:1}) and (\ref{eq:2}).

Lines~\ref{algline:xt2}-\ref{algline:charge} of \ocmrate calculates the procurement amount $\xt$ by taking into account rate constraints in the following cases:

(1)  \textit{Inactive output and input rate, Line}~\ref{algline:xt2}. In this case, $ [d(t)-\min\{b(t-1),\rho_d\}]^{+}\leq \hat{x}(t) \leq \rho_c +d(t)$, hence,  $x(t) = \hat{x}(t)$.

(2) \textit{Active output rate constraint, Line}~\ref{algline:discharge}. In this case, $\hat{x}(t)$ fails to satisfy the demand, so we set the actual procurement amount $x(t) = \left[d(t)-\min\{b(t-1),\rho_d\}\right]^{+}$.

(3) \textit{Active input rate constraint, Line}~\ref{algline:charge}. In this case, $\hat{x}(t)$ will be beyond the input rate of storage, hence $x(t) =\rho_c +d(t)$.

\subsubsection{Calculating the Reservation Price}
\label{sec:batmanrate_rp}
The final step is to update the reservation price $\xi_i$ for each virtual storage.
For cases (1) and (2), the reservation price for each virtual storage is updated similar to that of \ocm, i.e., $\min\{\xi_i,p(t)\}$ (Lines~\ref{algline:xi2} and~\ref{algline:xi3}).
For case (3), $\xi_i$ is updated as follows.
Let $p$ be the updated reservation price, whose value is the solution to the following equation
\begin{equation*}
\sum_{i=1}^{v_t}\left[G_{B_{i}}(p)-G_{B_{i}}(\xi_i)\right]^+=\rho_c +d(t).
\end{equation*}
\ocmrate solves the above equation for $p$ by calling the sub-procedure \texttt{CalRP} (Algorithm \ref{alg:alg_4}) in iterative manner, with parameter $\varepsilon_2$ as the stopping criteria ($\varepsilon_2$ can be arbitrarily close to 0).

\begin{algorithm}[!t]
\caption{\texttt{CalRP} ($d(t),\rho_c, v, \{\xi_i,B_i\}_{i=1:v},\varepsilon_2$)}
\label{alg:alg_4}
\begin{algorithmic}[1]
\State $p\leftarrow p_{\min}$, $p'\leftarrow p_{\max}/\alpha$;
\While{$|p'-p|>\varepsilon_2$}
\State $z \leftarrow \sum_{i=1}^{v}\left[G_{B_{i}}((p+p')/2)-G_{B_{i}}(\xi_i)\right]^+$
\State \textbf{if} $ z > \rho_c +d(t)$ \textbf{then} $p\leftarrow (p+p')/2$ \textbf{else} $p'\leftarrow (p+p')/2$
\EndWhile
\State \textbf{Output}: $p$
\end{algorithmic}
\end{algorithm}

\subsubsection{Competitive Analysis of \ocmrate}
Last but not the least, in this section, we briefly explain the proof sketch for the result in Theorem~\ref{thm:batmanrate}. The rigorous proof is given in Appendix~\ref{sec:batmanrate_analysis}.

The proof sketch is as follows. We first show that in worst case, the output rate is not active. Then, we find an upper bound on the cost of \ocmrate similar to that of in Lemma~\ref{lem:online_cost}. This upper bound gets affected by input rate constraints.
The rest is akin to the competitive analysis of \ocm. 
\section{Empirical Evaluation}
\label{sec:exp}
To illustrate the performance of our algorithms, we focus on the example of energy procurement, and evaluate their performance using various traces. Our results answers the following questions:

\paragraph{(1) How does the empirical cost ratio of \ocm compare to the theoretical competitive ratio?}
We find that \ocm achieves a significantly smaller average cost ratio than the worst-case competitive ratio guarantee provided by our theoretical analysis (Observation~\ref{obs:tcr}).
\paragraph{(2) How does \ocm compare to the existing algorithms?}
	We find that \ocm outperforms all the baseline and existing algorithms~\cite{chau2016cost,guo2012TPDS,urgaonkar2011optimal} by 7\%-15\%, on average. 
\paragraph{(3) How sensitive is \ocm to various parameters such as the penetration of renewable?} Our experiments demonstrate that \ocm is lightly affected by these parameters as compared to substantial performance fluctuations in alternative algorithms (Observations~\ref{obs:renew}).
\paragraph{(4) How does \ocmrate and compare to the basic \ocm?} We find that the empirical cost ratio of \ocmrate is only slightly worst than \ocm once rate constraints are tight (Observation~\ref{obs:batmanrate}).

\subsection{Experimental Methodology}
We perform trace-based simulations with extensive data traces.  
We use Akamai traces for the energy demand, four different electricity market for the energy prices, and nearby renewable generations to evaluate the results with more uncertainty from renewable generation. The details of data traces are given in Appendix~\ref{app:data}.

\begin{table}[!h]
	\caption{Summary of algorithms that are evaluated}
	\label{tbl:comp_alg}
	\vspace{-4mm}
	\footnotesize
	\begin{center}
		\begin{tabular}{|c |L{6.3cm}|}
			\hline
			\multicolumn{2}{|c|}{\textbf{Our proposed online algorithms}}\\\hline\hline
			\ocm & Basic online algorithm (\S\ref{sec:batman})\\\hline
			\ocmrate & Online algorithm with rate constraints (\S\ref{sec:batmanrate})\\\hline
			\preday & A simple data-driven approach to use optimal solution for the previous day for the current day\\\hline\hline
			
			\multicolumn{2}{|c|}{\textbf{Other algorithms for comparison}}\\\hline\hline
			{\tt OPT} & Optimal offline solution with storage\\
			\hline
			\bon & Optimal offline solution without storage\\
			\hline
						\fon~\cite{chau2016cost} & State-of-the-art online algorithm with fixed threshold price\\\hline
			\lyp~\cite{guo2012TPDS} & State-of-the-art online Lyapunov-based algorithm\\
			\hline
		\end{tabular}
		\normalsize
	\end{center}
\end{table}

\subsubsection{Comparison Algorithms}
\label{sec:compalg}
We implemented our algorithms and several other state-of-the-art algorithms for comparison as described below (see  Table~\ref{tbl:comp_alg}).

$\vartriangleright$ ({\tt OPT}) Optimal offline algorithm with storage by solving \prob in \S\ref{sec:formulation}. Since {\tt OPT} represents the best achievable cost for the given inputs, all other algorithms are evaluated by computing {\em empirical cost ratio} which is the ratio of the cost of the algorithm with the cost of {\tt OPT}. The cost ratio is always greater than equal to 1 and lower the cost ratio of an algorithm, the better the performance.

$\vartriangleright$ (\bon) A baseline scheme that simply satisfies the net energy demand from the grid, assuming no storage is available. The cost ratio of \bon quantifies the maximum benefit of having storage.

$\vartriangleright$ (\preday) Our data-driven approach that uses the optimal derived for the previous day for the current day by projecting into a feasible range (satisfying capacity, demand, and rate constraints). \preday is representative of a statistical approach that uses historical statistics to inform future decisions.

$\vartriangleright$ (\fon) Existing sub-optimal online algorithm~\cite{chau2016cost} is a simple strategy that uses a fixed threshold of $p = \sqrt{p_{\max}\times p_{\min}}$ as the purchasing threshold and fully charges the storage if the current price is lower than $p$, otherwise, discharges the storage as much as possible and purchases the remaining amount from the grid.

$\vartriangleright$ (\lyp) Lyapunov-based approach~\cite{guo2012TPDS,urgaonkar2011optimal} that uses Lyapunov optimization to solve \prob. Note that \cite{guo2012TPDS} considers load-balancing among multiple data centers as well, and to have a fair comparison, we focus on single data center model in \cite{guo2012TPDS}, and with this reduction both algorithms in \cite{guo2012TPDS,urgaonkar2011optimal} become similar.
\paragraph{Parameter Settings}
Unless otherwise mentioned, we set the length of each slot to $5$ minutes, according to FERC rule. The time horizon is 1 day, hence, $T=12\times 24 = 288$. We set the time horizon to 1 day to potentially see the impact of daily patterns in \preday.
The capacity of energy storage is set to $C= 18 \times \max_{t\in\mathcal{T}}d(t)$, sufficient to power the data center for 1.5 hours at max net demand.
In experiments with renewable, the renewable penetration for solar and wind is $50\%$. 
Finally, each data point in figures and tables corresponds to the average results of 30 runs (days) over a month, each with the corresponding demand, renewable generation, and market prices.

\subsection{\ocm vs. Alternative Algorithms}
\label{sec:micro}

In Table~\ref{tbl:comp}, the empirical cost ratio of 5 algorithms (\bon, \preday, \lyp, \fon, and \ocm) are reported across a broad set of settings: (i) four different locations; (ii) four different seasons; (iii) and with/without renewables. We report the notable observations:

\begin{table*}[!t]
	\caption{The empirical cost ratio of different algorithms in different markets and different seasons}
\vspace{-4mm}
	\label{tbl:comp}
	\centering
	\scriptsize
	\setlength\tabcolsep{4pt}
	\begin{tabular}{|c|c|c|c||c|c|c|c|c||c|c|c|c|c||c|c|c|c|c|}
		\hline
		&\multirow{2}{*}{\textbf{City/Market}} & \multirow{2}{*}{$\theta$} & \multirow{2}{*}{$\alpha$, Eq.~\eqref{eq:alpha}} & \multicolumn{5}{c|}{\textbf{Cost ratio for no renewables}} & \multicolumn{5}{c|}{\textbf{Cost ratio for 50\% wind penetration}} & \multicolumn{5}{c|}{\textbf{Cost ratio for 50\% solar penetration}}\\
		\hhline{~~~~---------------}
		&&& & \bon & \preday & \lyp & \fon &  \ocm  & \bon & \preday & \lyp & \fon &  \ocm  & \bon & \preday & \lyp & \fon &  \ocm  \\
		\hline \hline
		\multirow{5}{*}{\rotatebox[origin=c]{90}{ Winter}} & \textbf{Los Angles/CAISO} & 110.00 & 7.74 & 1.88 & 1.60 & 1.79 & 1.53 & 1.44 & 2.06 & 2.13 & 1.74 & 1.68 & 1.51 & 1.93 & 1.67 & 1.71 & 1.56 & 1.44 \\
		
		&\textbf{New York/NYISO} & 26.89 & 3.99 & 1.52 & 1.46 & 1.47 & 1.45 & 1.29 & 1.60 & 1.54 & 1.49 & 1.55 & 1.33 & 1.55 & 1.49 & 1.42 & 1.48 & 1.29 \\
		
		&\textbf{Dallas/ERCOT} & 15.83 & 3.13 & 1.23 & 1.15 & 1.13 & 1.13 & 1.15 & 1.26 & 1.16 & 1.15 & 1.16 & 1.13 & 1.26 & 1.17 & 1.16 & 1.15 & 1.13 \\
		
		&\textbf{Frankfurt/DE} & 2.22 & 1.36 & 1.11 & 1.03 & 1.10 & 1.09 & 1.07 & 1.13 & 1.04 & 1.12 & 1.10 & 1.08 & 1.11 & 1.03 & 1.10 & 1.09 & 1.07 \\
		\hline\hline
		&\textbf{Average} & \textbf{38.73} & \textbf{4.05} & \textbf{1.43} & \textbf{1.31} & \textbf{1.37} & \textbf{1.29} & \textbf{1.24} & \textbf{1.51} & \textbf{1.47} & \textbf{1.37} & \textbf{1.36} & \textbf{1.26} & \textbf{1.46} & \textbf{1.34} & \textbf{1.35} & \textbf{1.31} & \textbf{1.23} \\
		
		\hline\hline
		\multirow{5}{*}{\rotatebox[origin=c]{90}{ Spring}} & \textbf{Los Angles/CAISO} & 96.95 & 7.29 & 1.99 & 1.87 & 1.51 & 1.54 & 1.34 & 2.05 & 1.73 & 1.68 & 1.63 & 1.39 & 2.01 & 1.94 & 1.82 & 1.56 & 1.35 \\
		
		&\textbf{New York/NYISO} & 28.79 & 4.11 & 1.54 & 1.42 & 1.47 & 1.44 & 1.33 & 1.58 & 1.45 & 1.49 & 1.49 & 1.34 & 1.57 & 1.45 & 1.52 & 1.49 & 1.34 \\
		
		&\textbf{Dallas/ERCOT} & 10.09 & 2.56 & 1.27 & 1.17 & 1.22 & 1.15 & 1.15 & 1.33 & 1.21 & 1.27 & 1.19 & 1.17 & 1.29 & 1.16 & 1.20 & 1.17 & 1.14 \\
		
		&\textbf{Frankfurt/DE} & 2.04 & 1.31 & 1.07 & 1.03 & 1.07 & 1.07 & 1.05 & 1.08 & 1.03 & 1.07 & 1.08 & 1.05 & 1.08 & 1.03 & 1.07 & 1.08 & 1.05 \\
		\hline\hline
		&\textbf{Average} & \textbf{34.47} & \textbf{3.81} & \textbf{1.47} & \textbf{1.37} & \textbf{1.32} & \textbf{1.30} & \textbf{1.22} & \textbf{1.51} & \textbf{1.35} & \textbf{1.38} & \textbf{1.34} & \textbf{1.24} & \textbf{1.49} & \textbf{1.39} & \textbf{1.40} & \textbf{1.32} & \textbf{1.22} \\	
		
		\hline \hline
		\multirow{5}{*}{\rotatebox[origin=c]{90}{ Summer}} & \textbf{Los Angles/CAISO} & 25.10 & 3.86 & 1.56 & 1.36 & 1.51 & 1.41 & 1.32 & 1.56 & 1.37 & 1.51 & 1.42 & 1.33 & 1.60 & 1.39 & 1.54 & 1.44 & 1.34 \\
		
		&\textbf{New York/NYISO} & 19.96 & 3.48 & 1.33 & 1.26 & 1.31 & 1.34 & 1.24 & 1.35 & 1.29 & 1.28 & 1.38 & 1.26 & 1.35 & 1.27 & 1.30 & 1.38 & 1.24 \\
		
		&\textbf{Dallas/ERCOT} & 5.91 & 2.03 & 1.17 & 1.07 & 1.14 & 1.12 & 1.09 & 1.20 & 1.10 & 1.14 & 1.13 & 1.09 & 1.18 & 1.07 & 1.14 & 1.13 & 1.07 \\
		
		&\textbf{Frankfurt/DE} & 2.31 & 1.37 & 1.08 & 1.03 & 1.08 & 1.07 & 1.06 & 1.09 & 1.04 & 1.09 & 1.08 & 1.06 & 1.09 & 1.03 & 1.08 & 1.08 & 1.07 \\
		\hline\hline
		&\textbf{Average} & \textbf{13.32} & \textbf{2.68} & \textbf{1.29} & \textbf{1.18} & \textbf{1.26} & \textbf{1.23} & \textbf{1.18} & \textbf{1.30} & \textbf{1.20} & \textbf{1.25} & \textbf{1.25} & \textbf{1.19} & \textbf{1.30} & \textbf{1.19} & \textbf{1.27} & \textbf{1.25} & \textbf{1.18} \\\hline \hline
		
		\multirow{5}{*}{\rotatebox[origin=c]{90}{Fall}} & \textbf{Los Angles/CAISO} & 51.84 & 5.42 & 1.58 & 1.70 & 1.39 & 1.42 & 1.29 & 1.64 & 1.89 & 1.43 & 1.47 & 1.31 & 1.63 & 1.80 & 1.44 & 1.45 & 1.30 \\
		
		&\textbf{New York/NYISO} & 36.04 & 4.57 & 1.71 & 1.71 & 1.67 & 1.61 & 1.40 & 1.81 & 1.77 & 1.75 & 1.75 & 1.43 & 1.74 & 1.72 & 1.64 & 1.65 & 1.40 \\
		
		&\textbf{Dallas/ERCOT} & 7.26 & 2.22 & 1.22 & 1.12 & 1.20 & 1.13 & 1.08 & 1.23 & 1.15 & 1.16 & 1.14 & 1.09 & 1.23 & 1.12 & 1.16 & 1.13 & 1.08 \\
		
		&\textbf{Frankfurt/DE} & 2.12 & 1.33 & 1.12 & 1.05 & 1.11 & 1.09 & 1.09 & 1.15 & 1.08 & 1.12 & 1.11 & 1.10 & 1.12 & 1.06 & 1.12 & 1.09 & 1.09 \\
		\hline\hline
		&\textbf{Average} & \textbf{24.31} & \textbf{3.38} & \textbf{1.41} & \textbf{1.40} & \textbf{1.34} & \textbf{1.31} & \textbf{1.21} & \textbf{1.46} & \textbf{1.47} & \textbf{1.36} & \textbf{1.36} & \textbf{1.23} & \textbf{1.43} & \textbf{1.43} & \textbf{1.34} & \textbf{1.33} & \textbf{1.22} \\\hline\hline
		
		\multirow{5}{*}{\rotatebox[origin=c]{90}{Year}} & \textbf{Los Angles/CAISO} & 70.97 & 6.28 & 1.75 & 1.63 & 1.55 & 1.47 & 1.35 & 1.83 & 1.78 & 1.59 & 1.55 & 1.39 & 1.79 & 1.70 & 1.63 & 1.50 & 1.36 \\
		
		&\textbf{New York/NYISO} & 27.92 & 4.06 & 1.52 & 1.46 & 1.48 & 1.46 & 1.31 & 1.58 & 1.51 & 1.50 & 1.54 & 1.34 & 1.55 & 1.48 & 1.47 & 1.50 & 1.32 \\
		
		&\textbf{Dallas/ERCOT} & 9.77 & 2.53 & 1.22 & 1.13 & 1.17 & 1.13 & 1.12 & 1.25 & 1.15 & 1.18 & 1.16 & 1.12 & 1.24 & 1.13 & 1.17 & 1.14 & 1.11 \\
		
		&\textbf{Frankfurt/DE} & 2.17 & 1.34 & 1.09 & 1.04 & 1.09 & 1.08 & 1.07 & 1.11 & 1.05 & 1.10 & 1.09 & 1.07 & 1.10 & 1.04 & 1.09 & 1.08 & 1.07 \\
		\hline\hline
		&\textbf{Average} & \textbf{27.71} & \textbf{3.76} & \textbf{1.40} & \textbf{1.31} & \textbf{1.32} & \textbf{1.28} & \textbf{1.21} & \textbf{1.44} & \textbf{1.37} & \textbf{1.34} & \textbf{1.33} & \textbf{1.23} & \textbf{1.42} & \textbf{1.34} & \textbf{1.34} & \textbf{1.30} & \textbf{1.21} \\
		
		\hline	
	\end{tabular}
\end{table*}

\begin{table*}[!t]
	\caption{Comparison of different algorithms using different energy storage technologies}
	\vspace{-1mm}
	\label{tbl:comp_rate}
	\centering
	\scriptsize
	\setlength\tabcolsep{4pt}
	\begin{tabular}{|c|c||c|c|c|c|c||c|c|c|c|c||c|c|c|c|c|}
		\hline
		\multirow{2}{*}{\textbf{City/Market}} & \multirow{2}{*}{$\theta$} & \multicolumn{5}{c|}{\textbf{Lithium-Ion} ($\rho_c/B = 0.35$)} & \multicolumn{5}{c|}{\textbf{Lead-Acid} ($\rho_c/B = 0.2$)} & \multicolumn{5}{c|}{\textbf{Compressed Air Energy Storage} ($\rho_c/B = 0.05$)}\\
		\hhline{~~---------------}
		&&  \bon & \preday & \lyp & \fon &  \ocmrate  & \bon & \preday & \lyp & \fon &  \ocmrate  & \bon & \preday & \lyp & \fon &  \ocmrate  \\
		\hline \hline
		\textbf{Los Angles/CAISO} & 25.10 &  1.49  &  1.34  & 1.44 &   1.33  &  1.31 & 1.47 & 1.33 & 1.44 & 1.32 &1.32 &1.43  &  1.31  &  1.42  &  1.29  &  1.32 \\
		
		\textbf{New York/NYISO} & 19.96 & 1.27 &   1.21 &   1.24  &  1.24  & 1.21 &  1.25 & 1.19 &    1.22 &  1.22  &  1.20 &     1.23  &  1.17  &  1.20 & 1.19  &  1.18 \\
		
		\textbf{Dallas/ERCOT} & 7.26 & 1.14  &  1.10  &  1.11  &  1.13  &  1.07
		&  1.14   & 1.10  &  1.13  &  1.13  &  1.07 &   1.13  &  1.09  &  1.13  &  1.12  &  1.07 \\
		
		\textbf{Frankfurt/DE} & 2.12 &   1.07  &  1.03  &  1.05  &  1.08  &  1.05
		&   1.06  &  1.04   & 1.05  &  1.08  &  1.05 &  1.06  &  1.03  &  1.05  &  1.08 &   1.05 \\
		\hline\hline
		\textbf{Average} & \textbf{13.32} & \textbf{1.24} & \textbf{1.17} & \textbf{1.21} & \textbf{1.19} & \textbf{1.16} & \textbf{1.23} & \textbf{1.16} & \textbf{1.21} & \textbf{1.18} & \textbf{1.16} & \textbf{1.21} & \textbf{1.15}& \textbf{1.20} & \textbf{1.17} & \textbf{1.15}\\
		\hline
	\end{tabular}
\end{table*}



\begin{myObs}
	\label{obs:tcr}
\ocm achieves a significantly smaller average cost ratio than the worst-case competitive ratio guarantee provided by our theoretical analysis. \emph{The average theoretical competitive ratio ($\alpha$ in Equation~\eqref{eq:alpha}) for year-round experiments over four locations is 3.76, while the empirical cost ratios for \ocm are much smaller, i.e., 1.21 for no renewable and with solar, and 1.23 for wind.}
\end{myObs}

\begin{myObs}
	\label{obs:comp}
\ocm outperforms the alternative algorithms, when averaged across the entire year and all four locations, and with/without renewables and it is close to offline optimal \emph{\texttt{OPT}}. \emph{For example, in the case with wind as the renewable source, the last row of Table~\ref{tbl:comp} shows that \ocm outperforms \bon by $15\%$, \preday by $10\%$, \lyp by $8\%$, and \fon by $7.5\%$. Further, on average over the whole year \ocm achieves a cost ratio of 1.21 to 1.23, i.e., a cost that is within 21--23\% of the cost of {\tt OPT}.
However. there are a few settings that other algorithms outperform \ocm, e.g., \preday in Frankfurt/DE. The reason is
investigated in the next observation.}
\end{myObs}

\begin{myObs}
\preday is the best algorithm in settings with low price uncertainty and recurring daily price patterns.  \emph{To elaborate this observation, we need to further investigate the dynamics of the real-time prices in DE market. Figure~\ref{fig:market_prices} shows the real-time prices in three days in August 2017 for NYISO with high fluctuation ratio and DE Market with low price fluctuations. Once can see that in DE Market the prices do not fluctuate a lot and there is almost a regular daily pattern. This regular daily pattern is the key to \preday's good performance since it uses the previous day values to derive the procurement plan for today. However, the irregular pattern in NYISO and other markets (as shown in Figure~\ref{fig:market_prices}) motivates our general \ocm approach, since relying on the past information or stochastic modeling is less effective in these real-world markets. Predictably, \preday does not perform as well in CAISO and NYISO  markets in Table~\ref{tbl:comp_alg} (the cost ratio of 1.63 and 1.46 for \preday as compared to 1.35 and 1.31 for \ocm).}
\end{myObs}

\begin{myObs}
	\label{obs:renew}
With increased uncertainty due to renewable penetration, the performance of \ocm is robust, however, the performance of \preday degrades substantially. \emph{The performance of \ocm slightly degrades with injection of 50\% renewable. The performance of \preday, however, degrades substantially (e.g., from 1.31 to 1.37 for wind). To further elaborate this, in Figure~\ref{fig:penetration}, we compare the performance of \ocm and \preday in different seasons and locations as the penetration level varies. The result signifies the robust performance of \ocm and degradation of \preday with increased renewable penetration.}
\end{myObs}

\vspace{-1\baselineskip}
\begin{figure}[!h]
	\begin{center}
		\subfigure[New York]{\label{fig:penetration_4season_NY}
			\includegraphics[angle=0,scale=0.15]{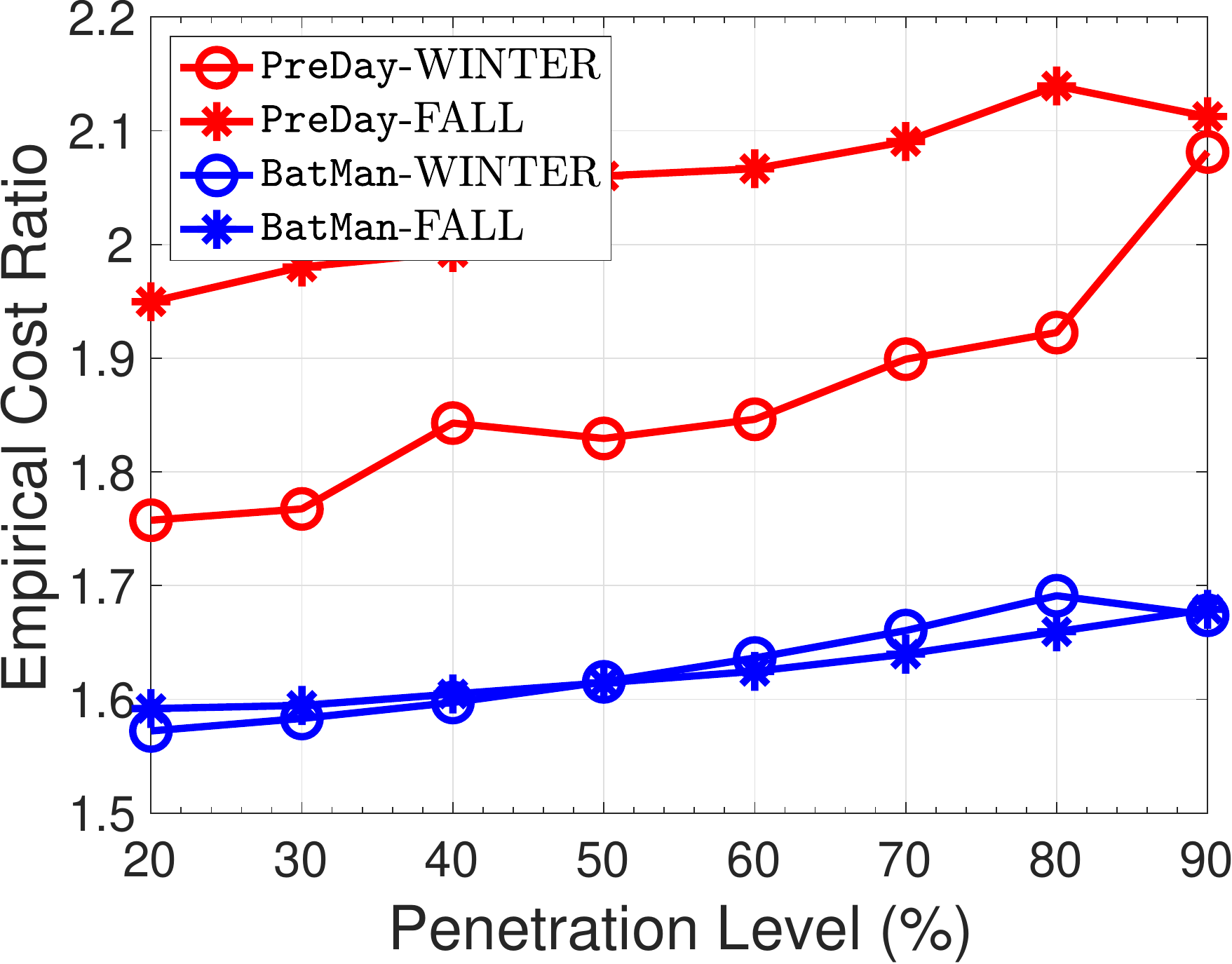}}\hspace{-0.5mm}
		\subfigure[Dallas]{\label{fig:penetration_4season_TX}
			\includegraphics[angle=0,scale=0.15]{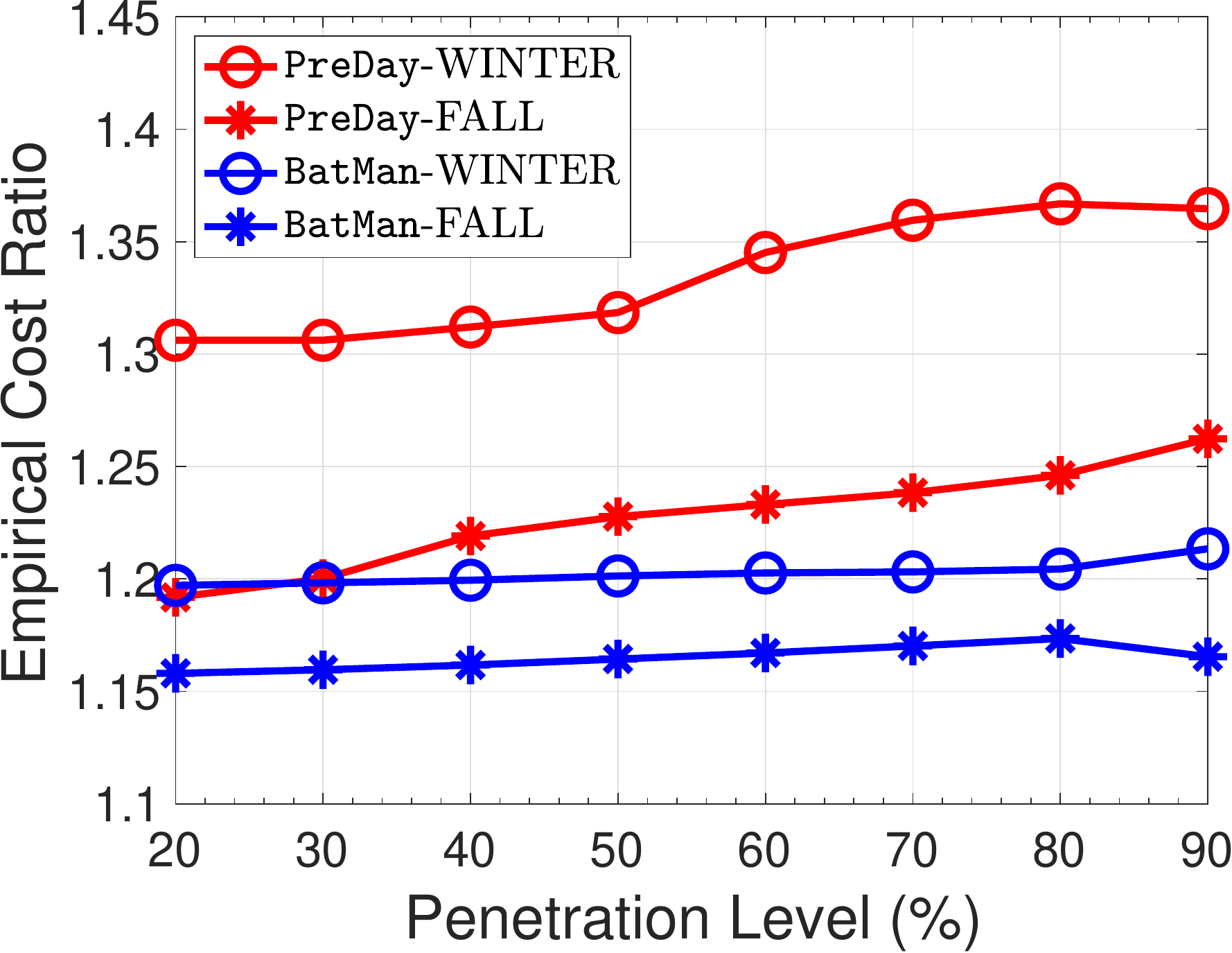}}\hspace{-0.5mm}
		\subfigure[Frankfurt]{\label{fig:penetration_4season_DE}
			\includegraphics[angle=0,scale=0.15]{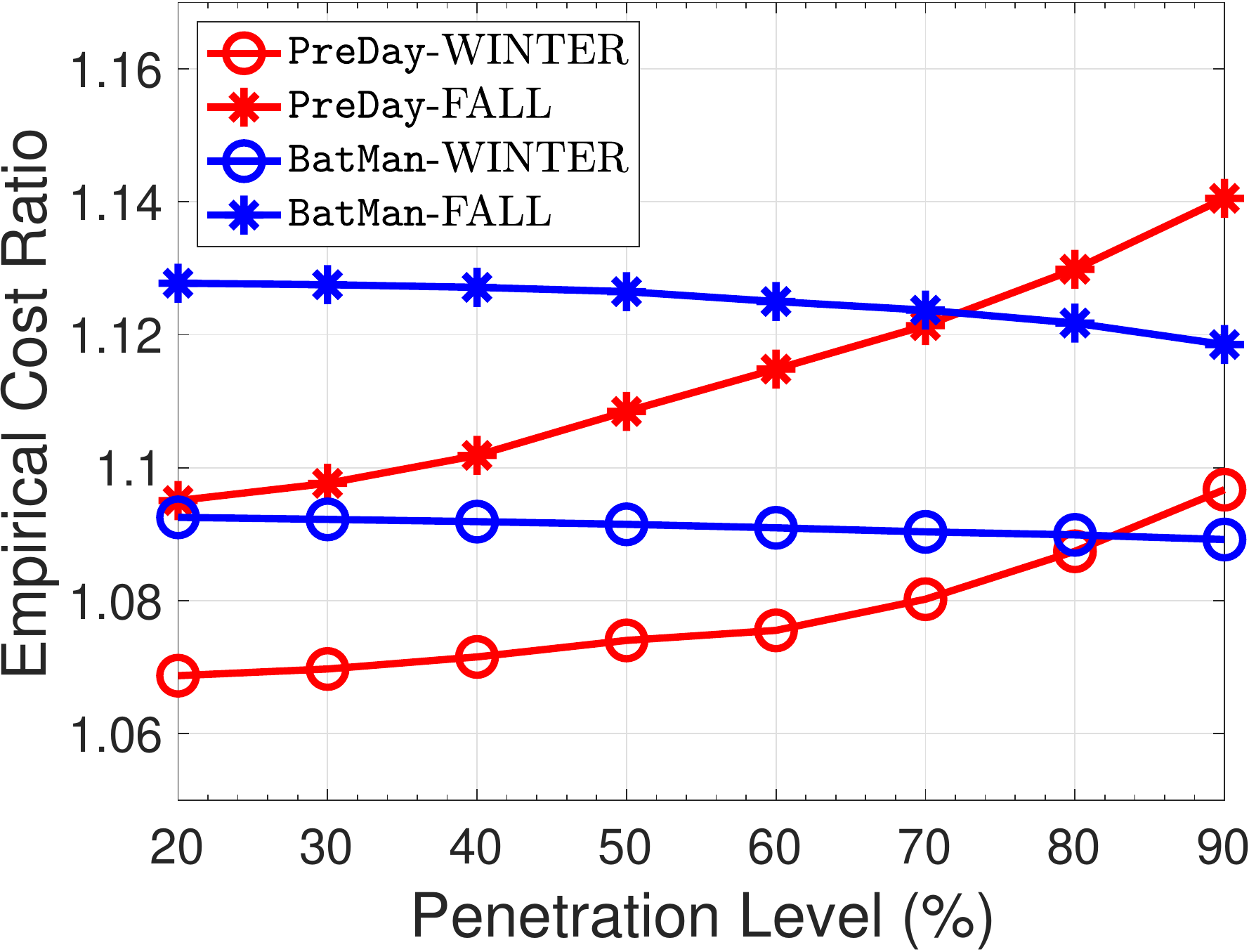}}\hspace{-0.5mm}
	\end{center}	
	\vspace{-5mm}
	\caption{The impact of renewable penetration}
	\label{fig:penetration}
\end{figure}
\vspace{-1\baselineskip}

\begin{myObs}
The seasonal and locational  patterns result in different degrees of uncertainty, thereby impact the performance of the algorithms substantially and increased uncertainty increases the cost ratio of algorithms. \emph{For example, in \ocm, the year-round cost ratio in Los Angles as the most uncertain location (highest $\theta = 70.97$, on average) is 1.47, while the same value for Frankfurt (lowest $\theta = 2.17$) as the least uncertain one is 1.07. As for seasonal variations, \ocm achieves the cost ratio of 1.18 in the summer as the least uncertain season (with average $\theta =13.32$), while this value is1.24 in the winter as the most uncertain scenario (with $\theta = 38.73$).}
\end{myObs}

\paragraph{Evaluation of \ocmrate}
\label{sec:exp_batmanerate}
To evaluate the performance of the second algorithm \ocmrate, we consider identical charging and discharge rates, i.e., $\rho_c = \rho_d$, and normalize it against the storage capacity. Hence, we define $ \rho = \rho_c/B$ as a measure of the rate at which an energy storage is charged/discharged relative to its capacity. 
A broad spectrum of storage technologies are integrated in data centers, each with a different $\rho$. To obtain practical values, we use the energy density as the normalized capacity, and power density as the normalized discharge rate from~\cite{Batteries}. We choose four common categories based on their $\rho$ values~\cite{chalise2015data,wang2012energy}: (1) Compressed Air Energy Storage with $\rho \approx 0.05$; (2) Lead-Acid with $\rho \approx 0.2$; (3) Lithium-Ion with $\rho \approx 0.35$, and (4) Flywheels $\rho \approx 1$. For Flywheels, there is no rate constraints and it reduces to \ocm. Hence,  we investigate the performance of \ocmrate for the first three technologies, and report the results in Table~\ref{tbl:comp_rate}.

\begin{myObs}
	\label{obs:batmanrate}
The performance of \ocmrate improves as $\rho$ increases, i.e., the rate constraints becomes more relax, while the performance of \preday exhibits no regular pattern. \emph{This observation is inferred from the results in Table~\ref{tbl:comp_rate} that reports the results for three representative energy storage technologies.}
\end{myObs}

%



\section{Conclusion}
We developed two competitive algorithms for online linear programming with inventory management constraints. We proved that both algorithms achieve the best possible competitive ratio. We evaluated the proposed algorithms using extensive data-traces from the application of energy procurement in data centers. 
As future work, we plan to extend the results and tackle maximization version of the problem, and general convex cost function.

\bibliographystyle{ieeetr}

\appendix

\section{Proof of Theorem~3.1}
\label{app:thm3.1}
To proof, we show that \ocm respects all the constraints in \probeqs.
First, it respects the demand covering constraint, i.e., $\xt \geq \dt - b(t-1)$, by the projection in Equation~\eqref{eq:xt}.
Second, we show that \ocm always respects the capacity constraints. At time slot $t$ which lies in a reservation period (see Definition~\ref{def:res_period}), the total amount of purchased asset from the beginning of current reservation period is equal to $\sum_{i=1}^{v}G_{B_i}(\xi_i)$, which is less than or equal to $\sum_{i=1}^{v}B_i$. The demand from the beginning of current reservation period is $\sum_{i=2}^{v}B_i$.
The asset stored in the physical storage is $\sum_{i=1}^{v}G_{B_i}(\xi_i)-\sum_{i=2}^{v}B_i$, which is less than or equal to $\sum_{i=1}^{v}B_i-\sum_{i=2}^{v}B_i=B_1$ according to the definition of $G_{B_i}(\xi_i)$.
$B_{1}$ is the capacity of the physical storage. That is, the amount of reserved power always respects the capacity constraint.

\section{Proofs Related to Analysis of \ocm}
\label{app:proof_batman}

\subsection{Proof of Lemma~\ref{lem:online_cost}}
\label{app:lem_upper_bound}
For each virtual storage, \ocm stores asset only if the market price is less than the reservation price. Hence, the cost of the stored assets in $j$-th storage is less than $\int_{0}^{\hat{b}_{i,j}}G^{-1}_{B_{i,j}}(b)db$.
By aggregation over $n$ reservation periods and virtual storages, we can compute the aggregate cost incurred by \ocm as
\begin{equation}
\label{eq:ub_1}
\sum_{i=1}^{n}\sum_{j=1}^{\hat{v}_i}\int_{0}^{\hat{b}_{i,j}}G^{-1}_{B_{i,j}}(b)db,
\end{equation}
where there are $\hat{v}_i$ virtual storage units at reservation period $i$,  $\hat{\xi}_{i,j}$ and $\hat{b}_{i,j}$ is the final storage level of virtual storage $j$ at reservation period $i$.
The additional amount of electricity needed to satisfy the demand in the idle period is equal to $D-\sum_{i=1}^{n}\sum_{j=1}^{\hat{v}_i}\hat{b}_{i,j}+\hat{b}$, where $D$ be the total demand during the time horizon, i.e., $D = \sum_{t\in\mathcal{T}} d(t)$, and $\hat{b}$ be the final storage level of the physical storage, i.e., $\hat{b} = b(T)$.
Hence, the cost of the online algorithm during the idle period is at most
\begin{equation}
	\label{eq:ub_2}
\left(D-\sum_{i=1}^{n}\sum_{j=1}^{\hat{v}_i}\hat{b}_{i,j}+\hat{b}\right) p_{\max},
\end{equation}
Adding~\eqref{eq:ub_1} and~\eqref{eq:ub_2} completes the proof.

\subsection{Proof of Lemma~\ref{lem:offline_cost}}
\label{app:offline_cost}
Similar to the online algorithm, the cost of an optimal offline solution, denoted as \textsf{cost}(\ofa),  can be also split  into two parts. To characterize a lower bound for the offline optimum, let us define $F_i(\beta)$ as the minimum cost of purchasing $\beta$ units of asset during the $i$-th reservation period.

Let $\beta_i$ be the asset purchased by the optimal offline solution during the $i$-th reservation period.
Obviously, the cost of the optimal offline solution during the $i$-th reservation period is at least $F_i(\beta_i)$. Let $\tilde{p}$ be the minimum price during the idle periods, then we have that the cost of the offline algorithm is lower bounded by $\sum_{i=1}^{n}F_i(\beta_i)+(D-\sum_{i=1}^{n} \beta_i)\times \tilde{p}$. This proof is complete.

\subsection{Additional Results Required to Prove Lemma~\ref{lem:neq}}
\label{app:neq}

We state the following lemma on the properties of function $G^{-1}_{B_i}(b)$ and $F_i(\beta)$ to facilitate the proof of Lemma~\ref{lem:neq}. Let $B_{i,j}$ be the capacity of the $j$-th storage and $\hat{b}_{i,j}$ be the final storage level at reservation period $i$.



\begin{lemma}
	\label{lem:opt_cost}
	Defining $\hat{b}_i$ as the initial state of the storage of the offline algorithm at $i$-th reservation period, there is a worst-case input instance such that for all $0<\beta<\beta'<\sum_{j=1}^{\hat{v}_i}B_{i,j}-\hat{b}_i$, we have
	\begin{enumerate}
		\item $F_i(0)=0$;
		\item $F_i(\beta')> F_i(\beta)$;
		\item $F_i(\beta')-F_i(\beta)\leq \frac{p_{\max}}{\alpha}(\beta'-\beta)$.
	\end{enumerate}
\end{lemma}
\begin{proof}
Statements (1) and (2) are straightforward.

Assume there is a worst instance $\omega = [\langle p(t),d(t)\rangle]_{t\in \mathcal{T}}$.
The adversary can construct a new instance $\omega'$ (as shown in Equation \eqref{eq:new_instance}) by adding one time slot before each time slot of instance $\omega$. Note that this is possible since the adversary can set the length of time horizon.
The market prices for the newly added time slots is $p_{\max}/\alpha$ and the demand is always equal to zero.

\begin{figure*}
	\begin{eqnarray}
	\label{eq:new_instance}
	\omega'=\left[\left\langle\frac{p_{\max}}{\alpha},0\right\rangle,\langle p(1),d(1)\rangle,\left\langle\frac{p_{\max}}{\alpha},0\right\rangle,\langle p(2),d(2)\rangle,\ldots ,\left\langle\frac{p_{\max}}{\alpha},0\right\rangle,\langle p(t),d(t)\rangle,\left\langle\frac{p_{\max}}{\alpha},0\right\rangle\right].
	\end{eqnarray}
	\hrule \hrulefill
\end{figure*}

In this way, we construct a new instance under which the cost of the online algorithm does not change, and that of the offline optimal solution will not increase. Thus, $\omega'$ is also the worst instance.
Let $F_{i}(\beta)$ be the minimum cost when buying $\beta$ units of asset during the $i$-th reservation period. When the optimal policy buys another $\beta'-\beta$, $\beta'<\sum_{j=1}^{\hat{v}_i}B_{i,j}-\hat{b}_i$, units of asset, the cost will not be larger than $\frac{p_{\max}}{\alpha}(\beta'-\beta)$, since the optimal policy can buy asset at any newly added time slots.
In this way we prove that, there is a worst instance, such that $F_{i}(\beta')-F_i(\beta)\leq p_{\max}/\alpha(\beta'-\beta),~\text{for}~0<\beta<\beta'<\sum_{j}\hat{b}_{i,j}-\hat{b}_i$.
This completes the proof.
\end{proof}

\section{Convergence of \texttt{InitVS}}
	\label{app:conv}
	\begin{theorem}
		Given a market price $p(t)\in [p_{\min},p_{\max}]$, \texttt{InitVS} converges to a feasible solution $B_{v}$ and $\hat{x}(t)$ which satisfy Equations (\ref{eq:1}) and (\ref{eq:2}) simultaneously.
	\end{theorem}
	\begin{proof}
		It is easy to see that $B'_{v}$ is always larger than or equal to $B_{v}$.
		And if the the value of $B'_{v}-B_{v}$ is larger than $\varepsilon_1$, the value of $B_{v}$ will increase by at least $\varepsilon_1$. Thus, there must be an iteration such that $B'_{v}-B_{v}\leq \varepsilon_1$.
	\end{proof}

\section{Competitive analysis of \ocmrate}
\label{sec:batmanrate_analysis}
If $D=0$, \ocmrate can be easily proved $\alpha$-competitive.
Thus, we focus our analysis on the case $D>0$.

Similar to the analysis for \ocm, we would like to upper bound the cost of \ocmrate. To achieve this, first we give the following two lemmas which characterize properties of the worst instance for \ocmrate.  Lemma \ref{lem:B_1} implies that in worst case, the output constraint is not active. Lemma~\ref{lem:online_cost_2} characterizes an upper bound on the cost of \ocmrate.

\begin{lemma}
\label{lem:B_1}
Under the worst case, $\tilde{x}(t)=0$, for $\forall t\in \mathcal{T}$, where $\tilde{x}(t) = [d(t)-\rho_d-\hat{x}(t)]^+$.
\end{lemma}
\begin{proof}
We prove this lemma by contradiction.
Assume there is a worst instance $\omega = [\langle p(t),d(t)\rangle]_{t\in \mathcal{T}}$, where $\tilde{x}(t)>0$ for time slot $t$.
We can construct a new instance which is the same as $\omega$ except at the $t$-th time slot.
For time slot $t$, the demand is set to $x(t)-\delta$, where $\delta<\tilde{x}(t)$, and the market price is equal to $p(t)$.
In this way, the cost of \ocm at time slot $t$ will decrease by $p(t)\delta$. The costs on other time slots are unchanged, because the modification on the demand does not influence the capacity and reservation price of virtual storage according to the rules of \ocmrate.
On the other hand, the cost of \ofa at time slot $t$ will decrease by at least $p(t)\delta$, since the procurement amount of \ofa is larger than $\delta$.
In this case, we have a new instance $\omega'$ under which the cost ratio is larger than that of the worst instance, contradicting the assumption.
This completes the proof.
\end{proof}

Lemma \ref{lem:B_1} implies that under the worst case the output constraint is not active.
Then, we take into account the influence of the input rate constraint.
Recall that in the basic version, \ocm, the procurement amount is always larger than or equal to $\hat{x}(t)$, which is computed in Equation (\ref{eq:1}).
With input constraint, $\hat{x}(t)$ may not be satisfied, and the maximum procurement amount is limited by $\rho_c+d(t)$.
We define $\mathcal{T}_{r}\subset \mathcal{T}$ be the set of time slots at which the input rate truncates the procurement amount. That is, the following equation holds for $t \in \mathcal{T}_r$.
\begin{equation*}
\sum_{i\leq v}\left[G_{B_{i}}(p(t))-G_{B_{i}}(\xi_i)\right]^+> \rho_c +d(t).
\end{equation*}
For $t\in \mathcal{T}_r$, we define $p'(t)$ as the value which satisfies the following equation.
\begin{equation*}
\sum_{i\leq v}\left[G_{B_{i}}(p'(t))-G_{B_{i}}(\xi_i)\right]^+= \rho_c +d(t),~\text{for}~\forall t \in \mathcal{T}.
\end{equation*}
$p'(t)$ is the actual reservation price computed in Algorithm \ref{alg:alg_4}, and obviously, $p'(t)>p(t)$. Let $\mu(t)=p'(t)-p(t)$ denotes the
difference between $p'(t)$ and $p(t)$.

Denote $x(t)$ and $x^*(t)$ as the amount of reserved asset by the online algorithm and the offline algorithm at time slot $t$, respectively.
Similar to Lemma \ref{lem:online_cost}, we have the following lemma which upper bounds the cost of the \ocmrate.
\begin{lemma}
	\label{lem:online_cost_2}
The cost of \ocmrate is upper bounded by
\begin{equation*}
\textsf{cost}(\ocmrate)\leq\mathcal{Q}+\left(D-\sum\limits_{i=1}^{n}\sum\limits_{j=1}^{\hat{v}_i}\hat{b}_{i,j}+\hat{b}\right)\cdot p_{\max}-\sum\limits_{t\in \mathcal{T}_r}\mu(t) x(t),
\end{equation*}
where $\mathcal{Q}=\sum\limits_{i=1}^{n}\sum\limits_{j=1}^{\hat{v}_i}\int_{0}^{\hat{b}_{i,j}}G^{-1}_{B_{i,j}}(b)db$.
\end{lemma}
\begin{proof}
By Lemma \ref{lem:B_1}, we have that, under the worst case, $\tilde{x}(t)=0$ and the amount of reserved asset is always less than or equal to the value computed in Equation (\ref{eq:1}).
Based on the analysis in \ref{lem:online_cost}, we have that the cost of \ocmrate is upper bounded by $\mathcal{Q}+\left(D-\sum\limits_{i=1}^{n}\sum\limits_{j=1}^{\hat{v}_i}\hat{b}_{i,j}+\hat{b}\right)\cdot p_{\max}$.
Moreover, $\forall t\in \mathcal{T}_r$, the actual price is less than the reservation price by $\mu(t)$, so the above upper bound is further modified to $\mathcal{Q}+\left(D-\sum\limits_{i=1}^{n}\sum\limits_{j=1}^{\hat{v}_i}\hat{b}_{i,j}+\hat{b}\right)\cdot p_{\max}-\sum\limits_{t\in \mathcal{T}_r}\mu(t) x(t)$.
This completes the proof.
\end{proof}

With the above two lemmas, the competitive ratio of \ocmrate is upper bounded by

\begin{equation*}
	\begin{split}
&\frac{\costocmrate-\hat{b}p_{\max}}{\costopt} \\
\leq&\frac{\mathcal{Q}+\left(D-\sum\limits_{i=1}^{n}\sum\limits_{j=1}^{\hat{v}_i}\hat{b}_{i,j}\right)\cdot p_{\max}-\sum\limits_{t\in \mathcal{T}_r}\mu(t) x(t)}{\sum\limits_{i=1}^{n}F'_i(\beta_i)+\left(D-\sum\limits_{i=1}^{n} \beta_i\right)\cdot \tilde{p}} \\
\leq & \max\left\{\frac{\mathcal{Q}+\left(\sum\limits_{i=1}^{n} \beta_i-\sum\limits_{i=1}^{n}\sum\limits_{j=1}^{\hat{v}_i}\hat{b}_{i,j}\right) \cdot p_{\max}-\sum\limits_{t\in \mathcal{T}_r}\mu(t) x(t)}{\sum\limits_{i=1}^{n}F'_i(\beta_i)},\alpha \right\} \\
\leq &\max\left\{\frac{\mathcal{Q}+\left(\sum\limits_{i=1}^{n}\sum\limits_{j=1}^{\hat{v}_i}B_{i,j}-\sum\limits_{i=1}^{n}\sum\limits_{j=1}^{\hat{v}_i}\hat{b}_{i,j}\right) p_{\max}-\sum\limits_{t\in \mathcal{T}_r}\mu(t) x(t)}{\sum\limits_{i=1}^{n}F'_i\left(\sum\limits_{j=1}^{\hat{v}_i}B_{i,j}\right)},\alpha \right\},
	\end{split}
\end{equation*}
where $F'_i(\beta)$ is defined as the minimum cost of purchasing $\beta$ units of asset during the $i$-th reservation period. The definition of $F'_i(\beta)$ is similar to that of $F_{i}(\beta)$ for for the basic version of the problem and it also respects the properties listed in Lemma \ref{lem:opt_cost}.

During the lifetime of the $j$-th virtual storage of the $i$-th reservation period,
the minimum reservation price is $\hat{\xi}_{i,j}$. 
The cost of the optimal algorithm satisfies

\begin{equation*}
\sum\limits_{i=1}^{n}F'_i\left(\sum_{j=1}^{\hat{v}_i}B_{i,j}\right)\geq \sum\limits_{i=1}^{n}\sum_{j=1}^{\hat{v}_i}\hat{\xi}_{i,j}B_{i,j}-\sum\limits_{t\in \mathcal{T}_r} \mu(t) x^*(t).
\end{equation*}

Then, we have

\begin{equation*}
\begin{split}
&\textsf{cr}(\ocmrate) \\
\leq &\max\left\{\frac{\mathcal{Q}+\left(\sum\limits_{i=1}^{n}\sum\limits_{j=1}^{\hat{v}_i}B_{i,j}-\sum\limits_{i=1}^{n}\sum\limits_{j=1}^{\hat{v}_i}\hat{b}_{i,j}\right) p_{\max}-\sum\limits_{t\in \mathcal{T}_r}\mu(t) x(t)}{\sum\limits_{i=1}^{n}\sum\limits_{j=1}^{\hat{v}_i}\hat{\xi}_{i,j}B_{i,j}-\sum\limits_{t\in \mathcal{T}_r} \mu_t x^*(t)},\alpha \right\} \\
\end{split}
\end{equation*}


The following lemma characterizes a bound on $x(t)/x^{*}(t)$ .

\begin{lemma}
Under the worst case, we have that $x(t)/x^{*}(t)$ is less than or equal to the competitive ratio, for any $t\in \mathcal{T}_{r}$.
\end{lemma}
\begin{proof}
Let $\omega = [\langle p(t),d(t)\rangle]_{t\in \mathcal{T}}$ be the worst instance and at time slot $t$, there is $x(t)/x^{*}(t)>\textsf{cr}(\ocmrate)$. We can construct a new instance $\omega'$ by increasing the market price at time slot $t$ by $\delta$, where $\delta\leq \mu(t)$. That is

\begin{equation*}
\omega'=[\langle p(1),d(1)\rangle,\ldots,\langle p(t)+\delta,d(t)\rangle,\ldots,\langle p(T),d(T)\rangle].
\end{equation*}

Under instance $\omega'$, the cost of \ocm will increase by $x(t)\delta$, and that of \ofa increase by less than $\frac{x(t)\delta}{\textsf{cr}(\ocmrate) }$.
In this way, we can get a worse instance $\omega'$ than $\omega$, contradicting the assumption that $\omega$ is the worst instance.
This completes the proof.
\end{proof}
By the above lemma, we have that,

\begin{equation*}
\frac{\sum\limits_{t\in \mathcal{T}_r}\tilde{p}_t x(t)}{\sum\limits_{t\in \mathcal{T}_r} \tilde{p}_t x^*(t)}\leq \textsf{cr}(\ocmrate).
\end{equation*}
Then, there is,

\begin{equation*}
\textsf{cr}(\ocmrate)\leq \max\left\{\frac{\mathcal{Q}+\left(\sum\limits_{i=1}^{n}\sum\limits_{j=1}^{\hat{v}_i}B_{i,j}-\sum\limits_{i=1}^{n}\sum\limits_{j=1}^{\hat{v}_i}\hat{b}_{i,j}\right) p_{\max}}{\sum\limits_{i=1}^{n}\sum\limits_{j=1}^{\hat{v}_i}\hat{\xi}_{i,j}B_{i,j}},\alpha \right\}.
\end{equation*}

Combining with Lemma \ref{lem:a_critical_property_on_g}, we have $\textsf{cr}(\ocmrate)\leq \alpha$.

In this way, we prove that \ocmrate is also $\alpha$-competitive.

\section{Energy Demand, Price, and Renewable Data Traces}
\label{app:data}
\paragraph{Data Center Energy Demand}
We use a repository of demand traces from Akamai's server clusters in several data centers collected during a 31 day period from multiple locations around the world. The data includes the server load information from 973 data centers in 102 countries, collected every 5 minutes. To calculate energy consumption as a function of load, we use the standard linear model~\cite{barroso2007case}. Let $d_{\mathrm{idle}}$ and $d_{\mathrm{peak}}$ be the energy consumption by an idle and a fully utilized server, respectively. Then, the energy (in kWh) consumed by a server serving normalized load $l \in [0,1]$ is $d(l) = d_{\mathrm{idle}} + (d_{\mathrm{peak}} - d_{\mathrm{idle}}) \times l$. In our experiments, we use $d_{\mathrm{idle}} = 100$kWh, and $d_{\mathrm{peak}}=250$kWh, representing energy proportionality factor, i.e., defined as $(d_{\mathrm{peak}} - d_{\mathrm{idle}})/d_{\mathrm{peak}}$, of $0.6$~\cite{palasamudram2012using}.
We report the results of different algorithms for a selection of data centers in the four different cities: Los Angeles, New York, Dallas, Frankfurt. A representative 7-days snapshot of the energy consumption is depicted in Figure~\ref{fig:dc_power_demand}.

\begin{figure}[!h]
	\begin{center}
		\subfigure[CAISO, high fluctuations]{\label{fig:CAISO_bp}
			\includegraphics[angle=0,scale=0.18]{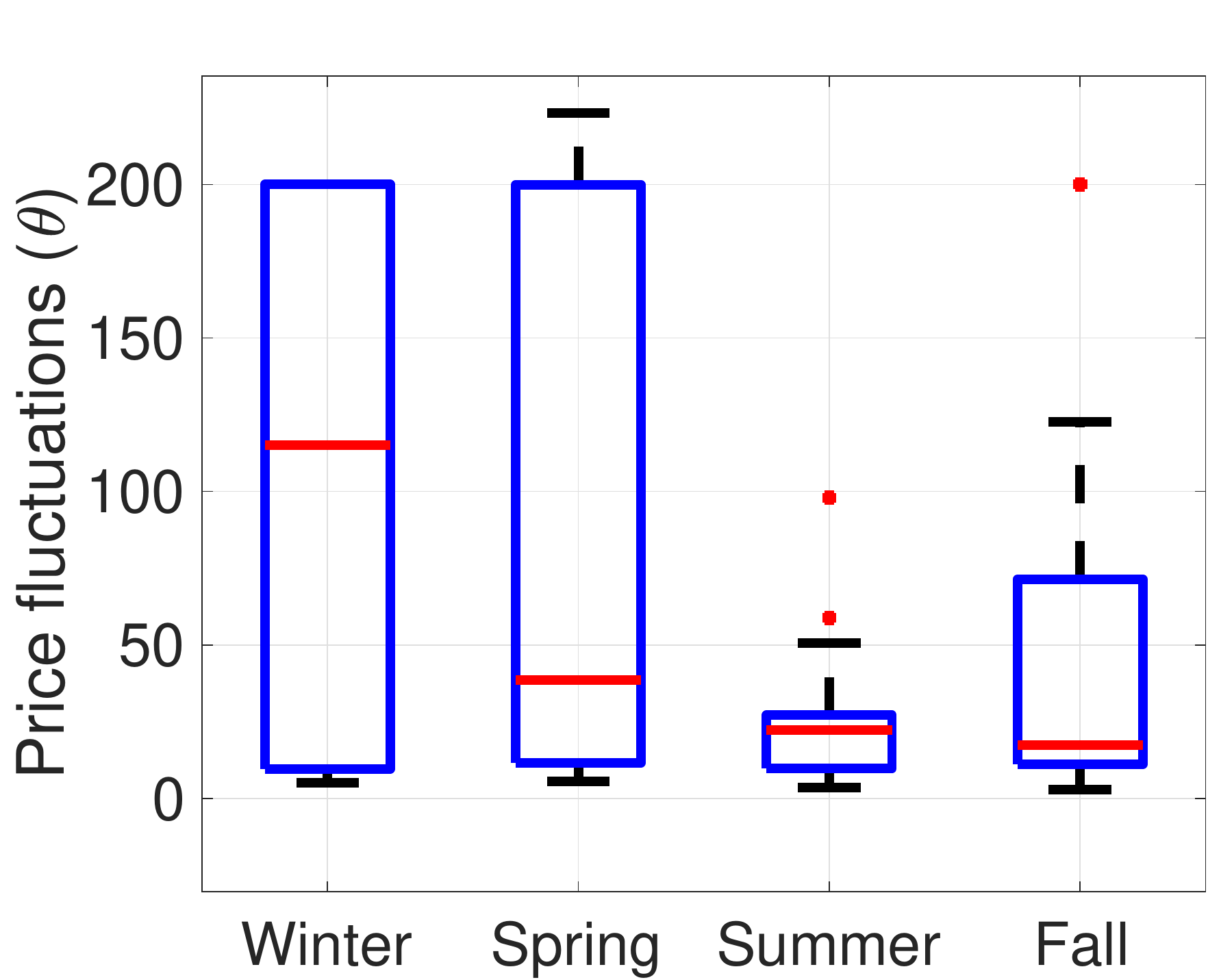}}\hspace{2mm}
		\subfigure[NYISO, high fluctuations]{\label{fig:NYISO_bp}
			\includegraphics[angle=0,scale=0.18]{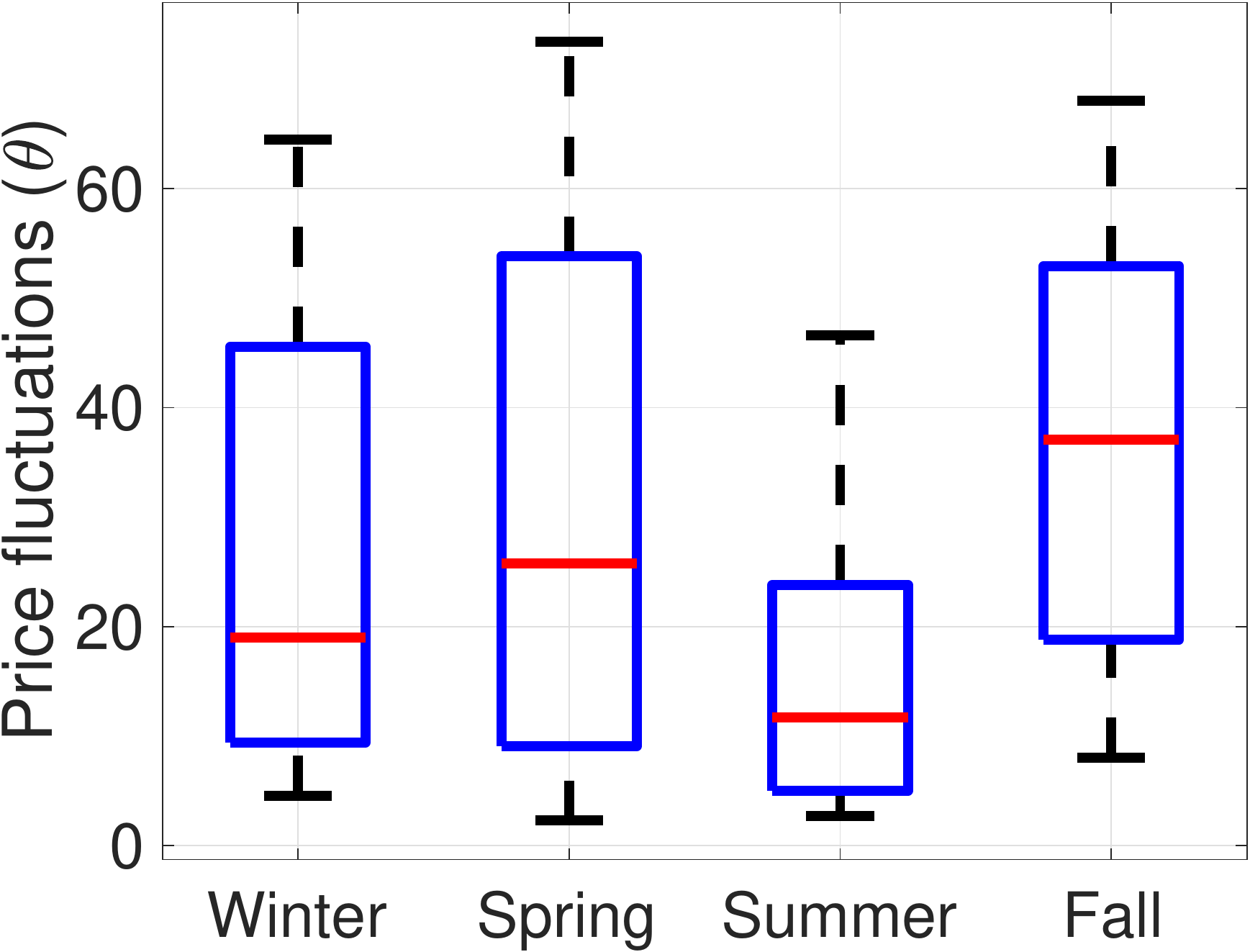}}\hspace{2mm}
		\subfigure[ERCOT, medium fluctuations]{\label{fig:ERCOT_bp}
			\includegraphics[angle=0,scale=0.18]{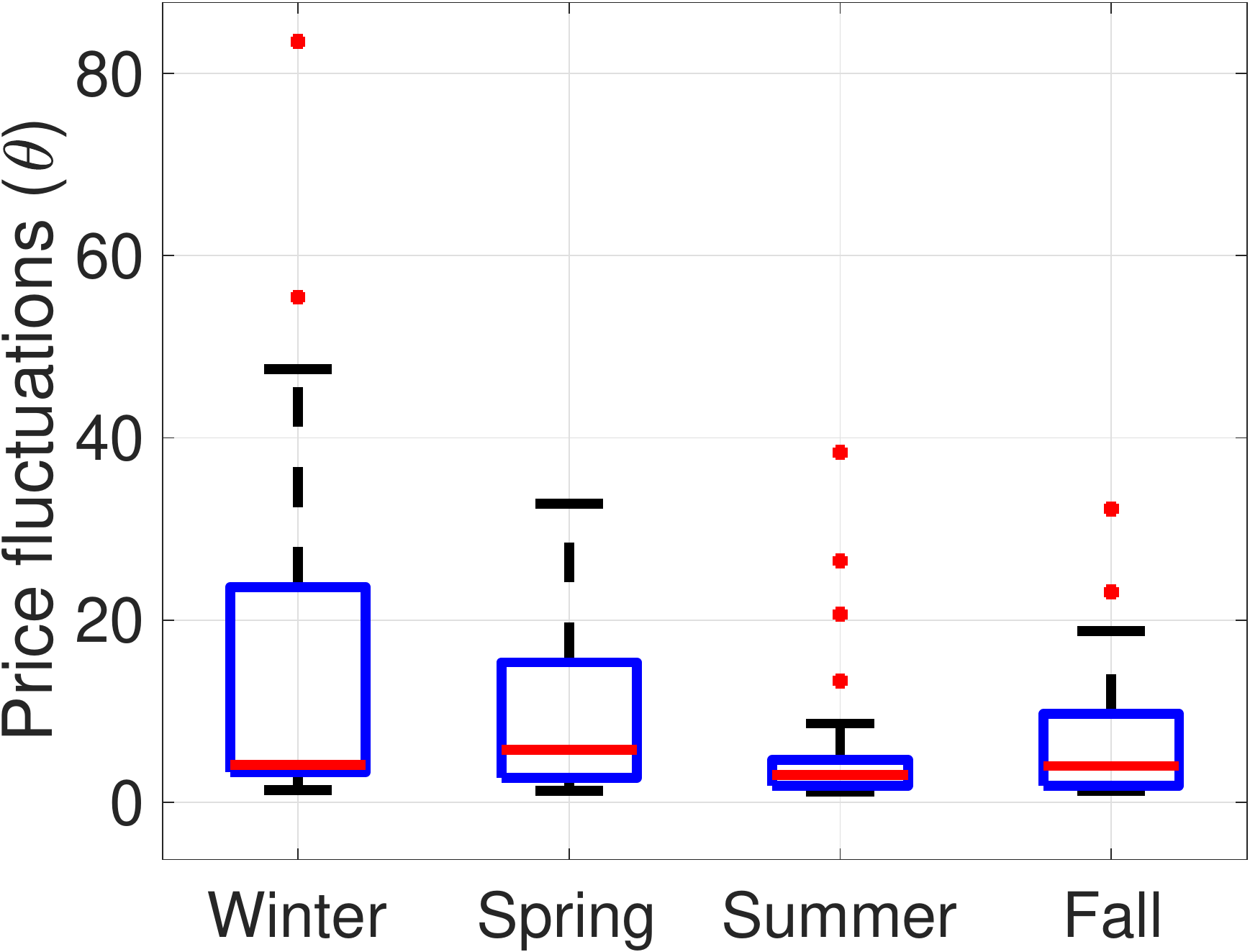}}\hspace{2mm}
		\subfigure[DE, low fluctuations]{\label{fig:DE_bp}
			\includegraphics[angle=0,scale=0.18]{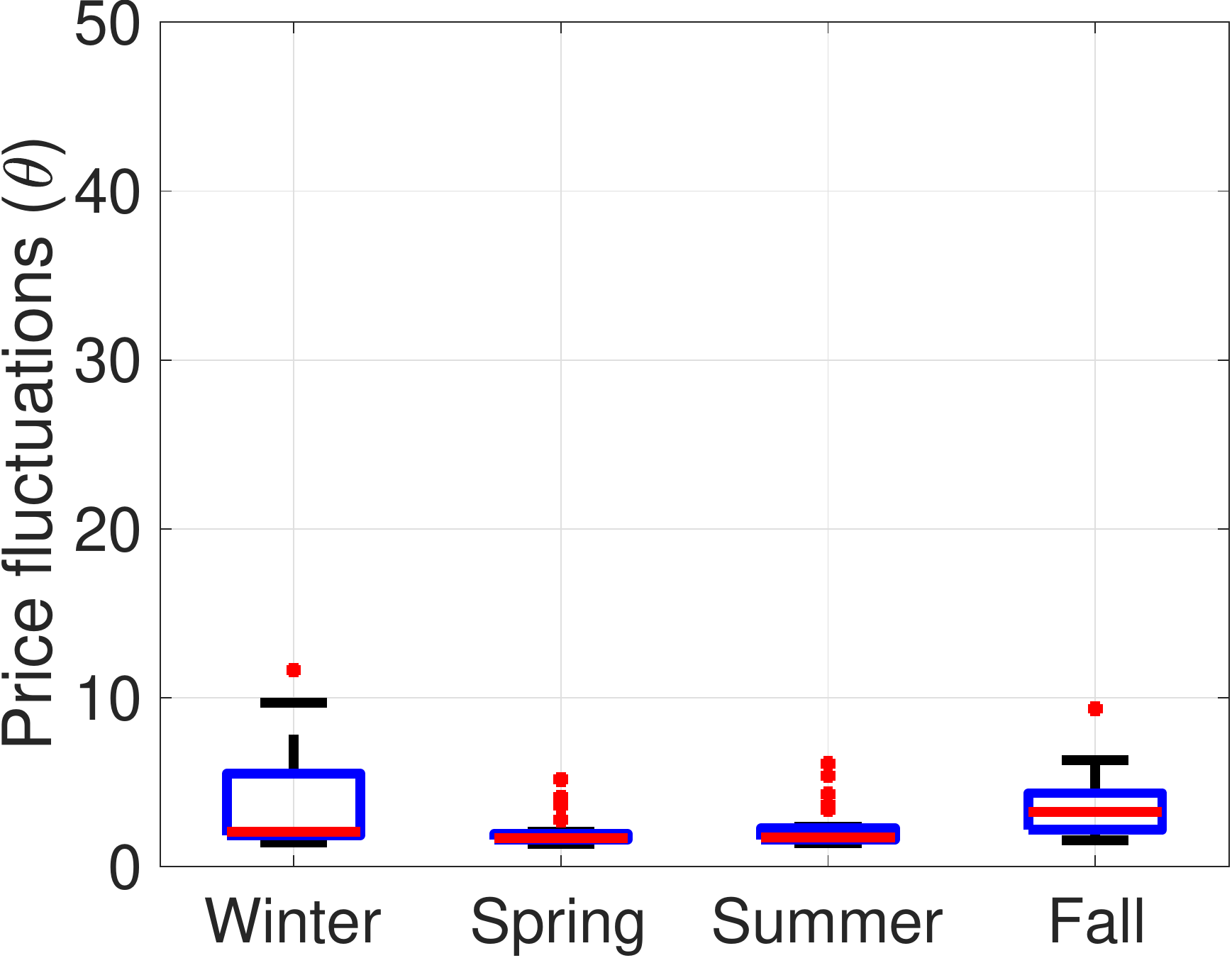}}\hspace{2mm}
	\end{center}	
	\vspace{-4mm}
	\caption{Price fluctuations in different seasons/locations}
	\label{fig:market_bp}
\end{figure}

\paragraph{Energy Prices}
We use the electricity prices from a local electricity market for each data center location, i.e., CAISO~\cite{CAISO} for Los Angeles, NYISO~\cite{NYISO} for New York, ERCOT~\cite{ERCOT} for Dallas, and German Electricity Market (abbreviated as DE in results) for Frankfurt.
Note that FERC is forcing the U.S. electricity markets to transition to real-time markets with 5-minutes settlement intervals~\cite{epri2016}. Currently CAISO and NYISO adapt this policy, and the rest are in the middle of this transition. To have a common settlement interval for all different markets, we set the length of each slot to 5 minutes, and for those that the current real-time market comes with different length (ERCOT with 15 minutes and DE with 1 hour intervals), we down-sample the market price readings to 5 minutes.

Recall that the performance of our algorithm is a function of parameter $\theta$ (see Equation~\eqref{eq:alpha}) as the price fluctuation ratio.
Note that different markets exhibit different price fluctuations in different seasons. In Figure~\ref{fig:market_bp}, the box plots for different markets in different seasons are shown. The results show that the fluctuations in spot prices in CAISO (Figure~\ref{fig:CAISO_bp}) and NYISO (Figure~\ref{fig:NYISO_bp}) are high, in ERCOT it is medium (Figure~\ref{fig:ERCOT_bp}), and in  DE it is low (Figure~\ref{fig:DE_bp}).
Consequently, to have a comprehensive experimental study in different fluctuation patterns, we compare the performance of different algorithms in different markets and different seasons.


\paragraph{Renewable Data Traces}
We evaluate the results of different algorithms in three different scenarios: (i) without any on-site renewable supply, (ii) with $50\%$ penetration on-site wind generation; and (iii) with $50\%$ penetration on-site solar generation. Note that with local renewable supply, the net energy demand, i.e., the total demand subtracted by the local renewable supply, must be procured from the grid with real-time pricing. Since the renewable supply is uncertain, the net demand in cases (ii) and (iii) will be more uncertain (as depicted in Figure~\ref{fig:LOSANGELES_wind}).

We use the solar data from PVWatts~\cite{dobos2014pvwatts} and obtain the hourly solar radiation in different seasons. We match each data center with solar readings from a location as close to it as possible. The exact distance from data center to the location from where the readings were obtained is show  in Table~\ref{tbl:solar_dist}. We scale the values such that $50\%$ of the total demand is satisfied by solar panels.
We set the parameters according to the default values~\cite[Table~2]{dobos2014pvwatts}. While the spot prices and energy demand readings are 5-minutes, the solar data is hourly. Hence, we make an assumption that the solar data is almost constant during each hour and use the hourly values for each 5 minute slots.
The wind traces for the U.S. locations are obtained from Eastern and Western data sets~\cite{wind}, and for European location are obtained from Open Power System Data~\cite{wind2}.
A summary of locations, markets, and distances to renewables is listed in Table~\ref{tbl:solar_dist}.
\begin{table}[!h]
	\caption{Summary of data center locations, markets, and nearby solar and wind power plants used in experiments}
	\vspace{-4mm}
	\label{tbl:solar_dist}
	\centering
	\footnotesize
	\begin{tabular}{|c|c|c|c|}
		\hline
		\textbf{City} & \textbf{Market} & \textbf{Dist. from solar} & \textbf{Dist. from wind}\\
		\hline \hline
		Los Angeles & CAISO & 80 mi. & 48 mi.\\
		\hline
		New York & NYISO & 37 mi. & 52 mi.\\
		\hline
		Dallas & ERCOT & 63 mi. & 145 mi.\\
		\hline
		Frankfurt & DE & 35 mi. & - \\
		\hline
	\end{tabular}
	\normalsize
\end{table}


\end{document}